\pdfoutput=1

\documentclass[a4paper,10pt]{article}

\newcommand{\longv}[1]{#1}
\newcommand{\shortv}[1]{}

\usepackage{microtype}

\usepackage{cmll}
\usepackage{xspace}
\usepackage{multicol}
\usepackage{amsfonts}
\usepackage{amsmath}
\usepackage{amsthm}
\usepackage[disable]
{todonotes}
\usepackage{hyperref}
\usepackage{etoolbox}
\usepackage{float}
\floatstyle{boxed} 
\restylefloat{figure}
\usepackage{framed}
\usepackage{txfonts}
\usepackage{graphicx}
\usepackage{caption}
\usepackage{amssymb}
\usepackage{bbold}
\usepackage{bussproofs}
\longv{\usepackage{subcaption}}

\usepackage{authblk}

\usepackage{pgf,tikz}
\usepackage{pgfplots}
\usetikzlibrary{matrix,arrows,positioning,backgrounds,decorations,shadows,petri}
\tikzset{
    >=stealth',
    pil/.style={
           ->,
           thick,
           shorten <=2pt,
           shorten >=2pt,}
}

\makeatletter
\newcommand\mmodif[1]{\ifinalign@ #1 \else \ifmmode #1 \else $#1$\xspace \fi \fi} 
\makeatother

\newcommand\at{@} 
\newcommand\+{\mathrm{+}}                
\newcommand\pt{\mathrm{\cdot}}           
\newcommand\dash{\mathrm{\text{-}}}

\newcommand\FV{\mmodif{\mathrm{FV}}}     
\newcommand\true{\mathbb{t}} 
\newcommand\false{\mathbb{f}} 

\newcommand\rta{\mmodif{\rightarrow}}
\newcommand\Rta{\mmodif{\Rightarrow}}

\newcommand\da{\mathrm{\downarrow}}
\newcommand\Da{\mathrm{\Downarrow^h}}
\newcommand\ua{\mathrm{\uparrow}}
\newcommand\Ua{\mathrm{\Uparrow^h}}
\newcommand\cons{\mathrm{\rta}}

\newcommand\ruledarrow[3]{\stackrel{#1}{\rta}\!\!{}_{#2}^{#3}}

\newcommand\Lra{\Leftrightarrow}

\newcommand\llb{\llbracket}
\newcommand\rrb{\rrbracket}

\newcommand\llc{(\! |}
\newcommand\rrc{|\! )}

\newcommand\I{\mmodif{\boldsymbol{I}}}

\newcommand\0{\mmodif{\boldsymbol{0}}}
\newcommand\eps{\mmodif{\boldsymbol{\epsilon}}}
\newcommand\bareps{\mmodif{\boldsymbol{\bar\epsilon}}}

\newcommand\Dinf{\mmodif{D_\infty}}
\newcommand\Pinf{\mmodif{P_{\!\infty}}}

\newcommand\BA{\boldsymbol{B}}
\newcommand\B[1]{\boldsymbol{B}(#1)}
\newcommand\dapprox[1]{\boldsymbol{ap}(#1)}
\newcommand\Nat{\mmodif{\mathbb{N}}}
\newcommand\ext[1]{\bold{ext}_{#1}}
\newcommand\Lamb{\mmodif{\Lambda}}
\newcommand\Lam[1]{\mmodif{\Lambda_{\tau,#1}}}

\newcommand\Lamcont[1]{\mmodif{\Lambda^{\!\protect\llc\pt\protect\rrc}_{\tau,#1}}}
\newcommand\Test[1]{\mmodif{\boldsymbol{T}_{\tau,#1}}}
\newcommand\Testc[1]{\mmodif{\boldsymbol{T}^{\protect\llc\pt\protect\rrc}_{\tau,#1}}}
\newcommand\Var{\mmodif{\protect\mathbb{Var}}}

\newcommand\Comp[1]{\mmodif{\bar #1}}               
\newcommand\Rea{\mathfrak{R}}                        
\newcommand\Reb{\mathfrak{Q}}                        

\newcommand\leobtau[1]{\mmodif{\sqsubseteq_{\tau(#1)}}}

\newcommand\equivobtau[1]{\mmodif{\equiv_{\tau(#1)}}}

\newcommand\Lcalcul{$\lambda$-calculus\xspace}

\newcommand\newsym[1]{\mmodif{#1}}
\newcommand\newsyminv[1]{}
\newcommand\newsyminvinv[2]{}
\newcommand\newsyminvsec[2]{}
\newcommand\newsymprem[2]{\mmodif{#1}}
\newcommand\newsymsec[2]{\mmodif{#2}}
\newcommand\newsyminvinvinv[3]{}
\newcommand\newdef[1]{{\em #1}}
\newcommand\newdefinv[1]{}
\newcommand\newdefinvinv[2]{}
\newcommand\newdefprem[2]{{\em #1}}
\newcommand\newdefpremsec[2]{{\em #1 #2}}
\newcommand\newdefsecpreminv[3]{{\em #2 #1}}
\newcommand\newdefinvinvinv[3]{}

\longv{
\newtheorem{definition}{Definition}
\newtheorem{theorem}[definition]{Theorem}
\newtheorem{example}[definition]{Example}
\newtheorem{lemma}[definition]{Lemma}
\newtheorem{remark}[definition]{Remark}

\newtheorem{conjecture}{Conjecture}
}
\newtheorem{proposition}{Proposition}

\shortv{\newtheorem{conjecture}[theorem]{Conjecture}}

\newenvironment{boite}{\vspace{-5pt}\begin{framed}\begin{center}\vspace{-10pt}}
                       {\end{center}\vspace{-10pt}\end{framed}\vspace{-5pt}}

\longv{
\title{Refining Properties of Filter Models:\quad\quad Sensibility, Approximability and Reducibility}
\author[1]{Flavien Breuvart}
\affil[1]{LIPN, UMR 7030, Univ Paris 13, Sorbonne Paris Cit\'e, France \\
  \texttt{breuvart@lipn.univ-paris13.fr}}
}

\begin{document}

\maketitle

\begin{abstract}
In this paper, we study the tedious link between the properties of sensibility and approximability of models of untyped $\lambda$-calculus. Approximability is known to be a slightly, but strictly stronger property that sensibility. However, we will see that so far, each and every (filter) model that have been proven sensible are in fact approximable. We explain this result as a weakness of the sole known approach of sensibility: the Tait reducibility candidates and its realizability variants.

In fact, we will reduce the approximability of a filter model $D$ for the $\lambda$-calculus to the sensibility of $D$ but for an extension of the $\lambda$-calculus that we call $\lambda$-calculus with $D$-tests. Then we show that traditional proofs of sensibility of $D$ for the $\lambda$-calculus are smoothly extendable for this $\lambda$-calculus with $D$-tests. 
\end{abstract}

\section*{Introduction}
%

{\bf Sensibility.}
It is the ability, for a model, to distinguish non terminating programs from meaningful ones by collapsing the interpretations of the formers (Def.~\ref{def:sens}). Through Curry-Howard isomorphism, it also corresponds to the consistence of the internal theory of the model. This shows the importance in understanding sensibility, but also the undecidability of such a property. 

Such profound but undecidable results are often targets for classification into a hierarchy of subclasses, serving as grinding stone for proof techniques. Here we take an unorthodox approach consisting in classifying sensible models by using as discriminator a slightly stronger property called ``approximability''. To our surprise, we found out that available methods to prove sensibility (reducibility) where not powerful enough to distinguish sensibility from approximability.

%

\smallskip \noindent
{\bf Approximability.} The approximation theorem (Def.~\ref{def:approx}) is an important concept when considering denotational models of the head reduction. In order to study head reduction, $\lambda$-calculists systematically use B\"ohm trees, which are basically normal forms of a degenerated $\lambda$-calculus using an error symbol (Def.~\ref{def:omegastuff}). Such objects are able to approximate terms, the same way as partial evaluations approximate the notion of evaluation. A model is approximable if the interpretation of a term is the limit of its finite B\"ohm approximants; i.e., infinite behaviors are, in the model, limits of finite ones. 

This notion has been extensively studied \cite[Section~III.17.3]{BarDekSta} and this article presents a new sufficient condition for approximability, the {\em weak positivity} by far encompassing any previous results on approximability (of filter models). As a property on models, approximability is supposed to be strictly stronger than sensibility. Indeed, approximability implies that the interpretation of any diverging terms (and only those) are collapsed into the interpretation of the error symbol $\Omega$. This inclusion is supposed to be strict as, for example, approximable models are not able to distinguish the Turing fixpoint from the Church fixpoint. In fact, there is a continuity of sensible but non-approximable $\lambda$-theories, it is surprising that we are not able to model any of those.

\longv{\smallskip} \noindent
{\bf Reducibility.} 
In this title, ``Reducibility'' refers to Tait reducibility methods~\cite{Tait67} and its modern extensions (including realisability). These methods used to prove structural properties of type systems and models, such as sensibility and approximability but also more practical properties~\cite{ThForFree}. For type systems, it consists of constructing saturated sets of terms with the wanted property by induction on types, and then in proving that every typable term has been included. For denotational models, the method is more subtle due to the structure not being inductive : one must find a fixpoint to be able to apply the method, but the fixpoint does not need to be computable or constructive in any way.
\shortv{\vspace{-0.1em}}

In Section~\ref{ssec:realisability}, we use the sensibility and the approximability as a grinding stone to perform yet a new dissection of those reducibility/realisability methods. We try to be as general as possible until the last moment in order to get the the coarsest possible characterization, but also in order to point over the specific weaknesses of the method. We will discuss in the conclusion and along the paper why we were not able to fill the gap between approximability and sensibility. In particular, we insist on the link between this obstacle and the difficulty to perform fixpoint on non-monotonous functions.

\longv{\smallskip} \noindent
{\bf Filter Models.}
Introduced in the 80’s using the notion of type as the elementary brick for their construction, filter models \cite{CDHL84} (Def.\ref{def:filterModels}) are extracted from a type theory with simple types enlarged by intersection types and subtyping. Formally, the interpretation of a $\lambda$-term is the filter generated by the set of its types. Variations on the intersection type theory induce different filter models. The resulting class essentially corresponds to the class of Scott complete lattices.  
\shortv{\vspace{-0.1em}}

Filter models (and domains) form one of the classes of models of untyped $\lambda$-calculus that have been the more broadly studied, but properties such as sensibility and approximability are yet to be understood perfectly.
In particular, a simple bibliographical analysis show that that the theoretically huge gap between sensible and approximable models have never been filed by any model. The best advancements toward this direction are covered by the third part of ``Lambda-calculus with types''  \cite{BarDekSta}.

\longv{\smallskip}  \noindent
{\bf $\lambda$-calculi with tests.} In order to exhibit the link between sensibility and approximability, we are using $\lambda$-calculi with tests of Section~\ref{sec:D-tests}. These are syntactic extensions of the untyped $\lambda$-calculus with operators defining types of the underlying intersection type system. We will see (Sec.~\ref{ssec:BTandTests}) that the approximability of a filter model $D$ is equivalent to the sensibility of the same model $D$ for the $\lambda$-calculus with $D$-tests $\Lam D$ (with respect to a notion of head convergence). This theorem brings together the notions of sensibility and approximability in a very novel way!
\shortv{\vspace{-0.1em}}

The calculi with tests played a central role in this paper. The idea of test mechanisms as syntactic extensions of the \Lcalcul
was first used by Bucciarelli {\em et al.} \cite{BCEM11} and developed further by the author \cite{Bre14,These,B.H*OpeAspect} for Krivine-models. The one presented in this paper is yet an other generalization to the broader (extensional and distributive) filter models. Originally inspired from Wadsworth's labeled $\lambda\bot$-calculus \cite{Wad76} and Girard experiments \cite{Gir87,DeFal00}, they are syntactic extensions of the $\lambda$-calculus with operators defining compact elements of the given models. Expressing the model in the syntax allows perform inductions directly on the reduction steps, rather than on the construction of B\"ohm trees.

\longv{\smallskip}  \noindent
{\bf Content.}
Section~\ref{ssec:preliminaries} will focus on preliminaries, with mostly standard presentations of the untyped lambda-calculus, the filter models and the B\"oms trees. In Section~\ref{sec:D-tests}, we present the $\lambda$-calculi with tests, mostly following previous works of the author~\cite{Bre14}; we give their syntax, their interpretation in filter models, and finally their main properties. Section~\ref{ssec:BTandTests} is short but central in this paper: we present here the collapse of the notions of approximability and sensibility at the level of test extensions.

In a Section~\ref{ssec:realisability}, we will present a standard proof of sensibility by reducibility adapted to $\lambda$-calculi with tests. Using our new equivalence between sensibility for this calculus with tests and approximability, this {\em a priori} standard proof of sensibility becomes a non-standard proof of approximability! This allows us to describe a condition for approximability that encompasses every known sensible extensional filter models, bringing these two properties closer than we believed them to be. 



\section{Preliminaries}
  \label{ssec:preliminaries}
  \subsection{The \Lcalcul}

In this paper, we only consider the minimal untyped $\lambda$-calculus with the contextual and/or the head reduction, in the pure tradition of Barendregt book~\cite{Barendregt}. 
$\lambda$-terms are defined up to $\alpha$-equivalence by the following grammar using notation ``{\em \`a la} Barendregt'' (where variables are denoted $x,y,z...$): 
\begin{center}
\begin{tabular}{l @{\hspace{2em}} c @{\hspace{4em}} c @{\ \ } c @{\quad} l}
($\lambda$-terms) & $\Lamb$: & $M,N$ & $::=$ & $x\quad |\quad \lambda x.M\quad |\quad M\ N$
\end{tabular}
\end{center}
We let $\FV(M)$ denote the set of free variables of a $\lambda$-term $M$. 
We let $M[N/x]$ denote the capture-free substitution of $x$ by $N$.
The $\lambda$-terms are subject to the $\beta$-reduction:
$$(\beta) \quad\quad (\lambda x.M)\ N\  \ruledarrow{\beta}{}{}\ M[N/x] $$
The writing $C\llc M\rrc$ denotes the term obtained by filling the holes of $C$ by $M$. 
The small step reduction~$\rta$ is the closure of $(\beta)$ by any context, and $\rta_h$ is the closure of $(\beta)$ by the rules:
\begin{center}
  \AxiomC{$M\rta_h M'$}
  \UnaryInfC{$\lambda x.M\rta_h \lambda x.M'$}
  \DisplayProof\hskip 10pt
  \AxiomC{$M\rta_h M'$}
  \AxiomC{\hspace{-3pt}$M$ is an application}
  \BinaryInfC{$M\ N\rta_h M'\ N$}
  \DisplayProof
\end{center}
The transitive reduction $\rta^*$ (resp $\rta_h^*$) is the reflexive transitive closure of $\rta$ (resp $\rta_h$). \\
The big step head reduction, denoted $M\Da N$, is $M\rta_h^*N$ for $N$ in a {\em head-normal form}, {\em i.e.}, of the form
\shortv{$\lambda x_1...x_k.y\ M_1\cdots M_k$\,, for $M_1,...,M_k$ any terms.}
\longv{$$\lambda x_1...x_k.y\ M_1\cdots M_k\ , \quad\quad \text{for }M_1,...,M_k \text{ any terms.}$$}
 We write $M\Da$ for the ({\em head}) {\em convergence}, {\em i.e.}, whenever there is $N$ such that $M\Da N$. We write $M\Ua$ for the divergence.

Other notions of convergence exist (strong, lazy, CbV...), but we focus on head convergence.

  \subsection{Filter Models} 
  \label{ssec:K-models}
%

We introduce here the main object of this article: distributive extensional filter models (DEFiM). 

Despite corresponding to reflexive complete lattices (endowed with continuous functions), we are not using this presentation to describe filter models, but rather its dual representation by Stone duality: the sup-lattice of compact elements. The following presentation is rather standard, and the notations can be find here~\cite{CaSa09} for example. This presentation has the advantage to match the representation of the interpretation of terms as intersection types derivations, as we will see in Proposition~\ref{prop:eqIT}.

The models consists of a set $D$ of ``types'' (or compact elements), and two operations: the intersection $\wedge$ (characterizing the induced order) and the functional arrow $\rta$ (characterizing the reflexive embedding). Moreover, we will consider extensionality, which means that the $\eta$-conversion is viable, it is enabled by (and is equivalent to) \longv{\todo{ref?} }the existence of a specific function $\ext{D}:D\rta\mathcal P_f(D{\times} D)$.

\begin{definition}[\cite{CDHL84}]\label{def:filterModels}
  A \newdef{filter model} is a triple $(D,\wedge,\cons)$ where:
  \begin{itemize}
  \item $D=(|D|,\wedge)$ is a pointed meet-semilattice, with $\omega$ and $\ge_D$ denoting top element and the order\shortv{.}\longv{:
    \begin{align*}
       \hspace{-2em}
      \alpha\wedge\alpha & = \alpha &
      \alpha\wedge\beta & = \beta\wedge\alpha &
      \alpha\wedge(\beta\wedge\gamma) & = (\alpha\wedge\beta)\wedge\gamma &
      \alpha\wedge\omega &= \alpha &
      (\alpha\ge_D \beta\ \Lra\ \alpha\wedge \beta = \beta)
    \end{align*}
    }
  \item \newsym{\protect \cons} is a binary operation on $D$ such that for any finite sequence $(\alpha_i,\beta_i)\in (D\times D)^n$: 
    $$\gamma\cons\delta\ge_D\bigwedge_i\alpha_i\cons\beta_i \quad\quad \Lra \quad\quad \delta\ge_D\bigwedge_{\{i\mid\gamma\le\alpha_i\}}\beta_i,$$
    in particular, $\gamma\cons\delta=\omega\ $ iff $\ \delta=\omega$.
    \end{itemize}
  A filter model is \emph{extensional} whenever there is a function $\ext{D}: D\cons \mathcal{P}_f(D\times D)$ that associates to each $\alpha\in D$ a finite subset $\ext{D}(\alpha)\subseteq D\times D$ such that:
     $$\alpha = \bigwedge_{(\beta,\gamma)\in \ext{D}(\alpha)} \beta\cons\gamma  $$
     \longv{
     It is free to consider that the image of $\ext{D}(\alpha)$ by $\rta$ is an anti-chain in the sens that for any pair $(\beta,\gamma)\in \ext{D}(\alpha)$ and any finite subset $I\subseteq \ext{D}(\alpha)$ with at least $2$ element:\todo{is it used?}
     $$ \bigwedge_{(\beta',\gamma')\in I}(\beta'\cons\gamma') \not\in Im(\cons) \quad\quad \text{and} \quad\quad \alpha\neq\bigwedge_{(\beta',\gamma')\in \ext{D}(\alpha)-(\beta,\gamma)}\beta'\cons\gamma' $$
     In particular $(\beta,\omega)\in \ext{D}(\alpha)$ implies $\alpha=\omega$, moreover $\ext D(\omega)=\{(\beta,\omega)\}$ for some arbitrary $\beta$ since $\beta\cons\omega = \omega$.\\
     }
     Unfortunately, the choice of the function $\ext{D}$ is generally not unique or even canonical. In order remove any influence from this choice, we restrict our study to \emph{distributive} filter models. A filter model $D$ is distributive whenever any $\alpha \ge \beta\wedge\gamma$ is accessible in the sens that there exists a decomposition $\alpha=\beta'\wedge\gamma'$ such that $\beta'\ge_D\beta$ and $\gamma'\ge_D\gamma$.
\end{definition}\longv{\todo{discussion on $\rta$ and reflexivity?}}

For short, we call \emph{DEFiM} the distributive extensional filter models.
By abuse of notation we may write \longv{the quadruple} $(D,\wedge,\rta,\ext{D})$ simply as $D$ when it is clear from the context that we are referring to a DEFiM. 



Creating a DEFiM from scratch is often heavy, as they have to satisfy complex rules even forcing the model to be an infinite object. Fortunately, there is a way to automatically infer the required properties from a smaller (often finite) core object. This core object is a partial DEFiM which is a basically a subset of a DEFiM. 

\begin{definition}\label{def:partialKweb}
  An \newdef{partial filter model} is a triple $(E,\wedge,\rta)$ satisfying the axioms of filter models except that $\rta$ is partially defined and for any $\alpha,(\beta_i)_{i\le n}\in E^{n+1}$:
  $$(\forall i\le n, \alpha\cons \beta_i\ \text{defined }) \quad\quad \Rta \quad\quad \alpha\cons \bigwedge_i\beta_i\ \ \text{defined as }\ \bigwedge_i(\alpha\cons\beta_i) $$
  It is a \newdef{partial DEFiM} if $\ext{E}$ is defined and $E$ satisfies the other axioms of DEFiMs.
\end{definition}

\shortv{
\begin{prop}
  Any partial DEFiM $E$ can be extensionaly completed into a DEFiM $\Comp{E}$ which is the smallest one containing $E$.
\end{prop}
\begin{proof}
   By recursively adding missing arrows and intersections.
\end{proof}
}
\longv{
\begin{definition}\label{def:Compl}
  The \newdefpremsec{completion}{of a partial DEFiM} $(E,\wedge,\rta,\ext{E})$ is the union 
  $$\Comp{E}\ :=\ \left(\:\bigcup_{n\in\mathbb{N}}E_n \:,\: \bigcup_{n\in\mathbb{N}}(\wedge_n) \:,\: \bigcup_{n\in\mathbb{N}}(\rta_n)\:,\:\bigcup_{n\in\mathbb{N}}\ext{E_n}\:\right)$$
  of partial completions $(E_n,\wedge_n,\rta_n,\ext{E_n})$ that are partial DEFiM defined by induction on $n$:\\
  The initialization $(E_0,\wedge_0,\rta_0,\ext{E_0}):=(E,\wedge,\rta,\ext{E})$ is performed by the partial DEFiM, and we continue by completing:
  \begin{itemize}
  \item $|E'_{n+1}| := \mathcal{P}_f(|E_n|\uplus (|E_n|^2{-}Dom(\rta_n)))$,
    for readability, use $a,b..$ for elements of $|E'_{n+1}|$ and we write $\alpha\cons_{\!*}\beta$ for $(\alpha,\beta)$ in the second component,
  \item $\cons'_{n+1}$ is defined only over $|E_n|^2\subseteq |E'_{n+1}|^2$ by $\{\alpha\} \cons'_{n+1}\{\beta\} := \{\alpha \cons_n\beta\}$ whenever $(\alpha,\beta)\in Dom(\cons_n)$ and by $\{\alpha\} \cons'_{n+1}\{\beta\} := \{\alpha\cons_*\beta\}$ whenever $(\alpha,\beta)\in|E_n|^2-Dom(\cons_n)$,
  \item $\ext{n+1}'$ is defined over $|E'_{n+1}|$ by $\ext{n+1}'(a)=\{\ext{n}(\alpha)\mid \alpha\in E_n\cap a\}\wedge \{(\alpha,\beta) \mid \alpha\cons_*\beta \in a \}$.
  \item $|E_{n+1}| := |E'_{n+1}|/_{\equiv}$ is the quotient of $|E'_{n+1}|$ by the equivalence $a\equiv b$ whenever:
    $$
      \forall (\alpha,\gamma)\in \ext{n+1}'(a),\quad \gamma \ge\!\! \bigwedge_{\{(\beta,\delta)\in \ext{n+1}(b)\mid \alpha\le \beta\}}\!\!\!\!\!\!\delta \quad\quad\quad\quad\quad
      \forall (\beta,\delta)\in \ext{n+1}'(b),\quad \delta \ge\!\! \bigwedge_{\{(\alpha,\gamma)\in \ext{n+1}(a)\mid \beta\le \alpha\}}\!\!\!\!\!\!\gamma
    $$
  \item $\wedge_{n+1}$, $\cons_{n+1}$ and $\ext{n+1}$ are the quotients of $\wedge$, $\cons'_{n+1}$ and $\ext{n+1}'$ by $\equiv$ (notice that $\cons'_{n+1}$ only need to be defined for one element equivalent class for $\cons'_{n+1}$ to be defined).
  \end{itemize}
  We consider that $E_n\subseteq E_{n+1}$ since for each $\alpha\in |E_n|$, $\{\alpha\}$ is in a different equivalence class.
\end{definition}
}

\longv{
\begin{remark}\label{rk:completion}
  The completion of a partial filter model $(E,\rta,\ext{E})$ is well defined and corresponds to the coarsest DEFiM $\Comp{E}$ containing $E$. In particular, any DEFiM model $D$ is the completion of itself: $D=\Comp{D}$.
\end{remark}
}

\begin{example}\label{example:1}
  Most filter models found in the literature can in fact be given as extensional completions of extremely simple partial filter models. Here are some example, the three first one are from the literature and the two last one are fully expressing the power of the extensional completion:
  \begin{enumerate}
  \item \label{enum:Dinf} \newdef{Scott's $\protect\Dinf$} \newsyminvsec{K-model}{\protect\Dinf} \cite{Scott} is the completion of \vspace{-0.3em}
    \begin{align*}
      |D| &:= \{\omega,*\},   &   \omega \wedge * &:= *,   &   \omega \cons * &:= * &  \ext{D}(*):= \{(\omega,*)\}.
    \end{align*}
    Notice, that $*\cons *$ is undefined in $D$ so that we need the completion.
  \item \label{enum:Park} \newdef{Park's $\protect\Pinf$} \newsyminvsec{K-model}{\protect\Pinf} \cite{Par76} is the completion of\vspace{-0.3em}
     \begin{align*}
      |P| &:= \{\omega,*\},   &   \omega\wedge * &:= *,   &  *\cons * &:= * &  \ext{P}(*) &:= \{(*,*)\}.
    \end{align*}
  \item \label{enum:D*} \newsymsec{K-model}{Norm} or \newsymsec{K-model}{\protect D^*_\infty} \cite{CDZ87} is the completion of\vspace{-0.3em}
    \begin{align*}
      |D^*| &:= \{\omega,p,q\},   &   \omega\wedge p &:= p & \omega\wedge q&:= q & p\wedge q &:= q\\
       p\cons q &:= q & q\cons p &:= p & \ext{D^*}(q)&:=\{(p,q)\} &\ext{D^*}(p)&:=\{(q,p)\},
    \end{align*}
  \item \label{enum:Z} ${Z_\infty}$ is the completion of \vspace{-0.3em}
    \begin{align*}
      |Z| &:= \{\underline n\mid n\ge 0\},   &   n \wedge \omega &:= n,   &   \omega \cons \underline {n\+1} &:= \underline n &  \ext{D}(n):= \{(\omega,\underline {n\+1})\}.
    \end{align*}
  \item \label{enum:U} ${U_\infty}$ is the completion of \vspace{-0.3em}
    \begin{align*}
      |U| &:= \{\underline n\mid n\ge 0\},   &   n \wedge \omega &:= n,   &   \underline {n\+1} \cons \underline {n\+1} &:= \underline n &  \ext{D}(n):= \{(\underline {n\+1},\underline {n\+1})\}.
    \end{align*}
  \end{enumerate}
\end{example}

\longv{
\begin{remark}\label{rk:completion}
  The completion of a partial filter model is in fact the free completion in the sens that for any partial DEFiM $E\subseteq D$ contains in a DEFiM $D$, there is a function $\phi: \Comp E\rta D$ stable in $E$ such that $\llb .\rrb_{\Comp E} \subseteq \llb .\rrb_{D}$, where $\llb .\rrb_{\Comp E}$ (resp. $\llb .\rrb_{D}$) is the interpretation of the $\lambda$-calculus into $\Comp E$ (resp. $D$) as defined below.
\end{remark}
}

Filter models where introduced so that the interpretation of the \Lcalcul into a given $D$ can be equivalently characterized by a specific {\em intersection type system}, whose types are elements $\alpha\in D$ and with $\wedge$ modeling the intersection and $\cons$ the logical implication.

\longv{
\begin{figure*}
  \begin{center}
    $\llb x_i \rrb_D^{\vec x}  =  \{(\vec \alpha,\beta)\ |\ \beta\ge\alpha_i\}$
    \hspace{5em}
    $ \llb \lambda y.M \rrb_D^{\vec x}  =  \{(\vec \alpha,\bigwedge_i(\beta_i\cons\gamma_i))\ |\ \forall i,\ (\vec \alpha\beta_i,\gamma_i)\in\llb M\rrb_D^{\vec xy}\}$
    \vspace{0.3em}\\
    $\!\llb M\ N \rrb_D^{\vec x}  =  \{(\vec \alpha,\bigwedge_i\beta_i)\ |\ \exists \vec\gamma_i,(\vec \alpha,\bigwedge_i(\gamma_i\cons\beta_i))\in\llb M\rrb_D^{\vec x}\ \wedge (\vec \alpha,\bigwedge_i\gamma_i)\in\llb N\rrb_D^{\vec x}\}$
  \end{center}
  \caption{Direct interpretation of $\Lamb$ in the model $D$}\label{fig:intLam}
\end{figure*}
\begin{figure*}
  \begin{center}
    \AxiomC{$\phantom{M}$}
    \UnaryInfC{$x:\alpha\vdash x:\alpha$}
    \DisplayProof\hskip 50pt
    \AxiomC{$\Gamma\vdash M:\alpha$}
    \UnaryInfC{$\Gamma,x:\beta\vdash M:\alpha$}
    \DisplayProof\hskip 50pt
    \AxiomC{$\Gamma\vdash M:\beta$}
    \AxiomC{$\alpha\ge\beta$}
    \BinaryInfC{$\Gamma\vdash M:\alpha$}
    \DisplayProof\\[0.5em]
    \AxiomC{$\Gamma,x:\alpha\vdash M:\beta$}
    \UnaryInfC{$\Gamma\vdash \lambda x.M:\alpha\cons\beta$}
    \DisplayProof\hskip 30pt
    \AxiomC{$\Gamma\vdash M:\alpha\cons\beta$}
    \AxiomC{$\Gamma\vdash N:\alpha$}
    \BinaryInfC{$\Gamma\vdash M\ N:\beta$}
    \DisplayProof\hskip 30pt
    \AxiomC{$\Gamma\vdash M:\alpha$}
    \AxiomC{$\Gamma\vdash M:\beta$}
    \BinaryInfC{$\Gamma\vdash M:\alpha\wedge\beta$}
    \DisplayProof
  \end{center}
  \caption{Intersection types for the $\lambda$-calculus in $D$} \label{fig:tyLam1}
\end{figure*}
}
\begin{definition}[Interpretation of $\lambda$-terms]
  In Figure~\ref{fig:intLam}, we give the interpretation of $M$ into a filter model $D$. The interpretation $\llb M\rrb^{x_1...x_n}_{D}$ of $M$ is suppose to be a morphism (Scott-continuous function) from $D^{\vec x}$ to $D$ where $\vec x$ is a superset of the free variables of $M$. Concretely, we use the Cartesian closedness of the underlying domain category do define $\llb M\rrb^{x_1...x_n}_{D}$ as a downward-close subsets of $(D^{op})^{\vec x}\times D$.

  In Figure~\ref{fig:tyLam1}, we give the intersection-type assignment corresponding to $D$. Notice that we can infer typing sequents for the form $\Gamma\vdash M:\alpha$ for $\Gamma=(x_1:\alpha_1,...,x_n:\alpha_n)$ an environment defined (at least) over all free variables of $M$.
\end{definition}

\begin{example}
  \vspace{-0.5em}
  \begin{align*}
    \llb\lambda x.y\rrb_{D}^y &= \left\{((\alpha), \bigwedge_i(\beta_i\cons\alpha'_i))\ \middle|\  \forall i,\ \alpha'_i\ge_{D}\alpha\right\},
    \quad \llb\lambda x.x\rrb_{D}^y = \left\{((\alpha), \bigwedge_i(\beta_i\cons\beta'_i))\ \middle|\  \forall i,\ \beta'_i\ge_{D}\beta_i\right\},\\
    \llb\I\rrb_{D} &= \left\{\bigwedge_i(\alpha_i\cons\alpha'_i)\ \middle|\  \forall i,\ \alpha'_i\ge_{D}\alpha_i\right\},\\
    \llb \underline{1}\rrb_{D} &=  \left\{ \bigwedge_i(\alpha_i\cons\alpha'_i)\ \middle|\ \exists \vec \beta',\vec\gamma',\vec \beta,\vec\gamma,
      \begin{matrix} 
        &\bigwedge_i\alpha'_i=\bigwedge_j(\beta'_j\cons\gamma'_j),
        &\bigwedge_j\gamma'_j = \bigwedge_k\gamma'_k, \\
        &\bigwedge_i\alpha_i\le\bigwedge_k(\beta_k\cons\gamma_k),
        &\bigwedge_j\beta'_j\le\bigwedge_k\beta_k,\\
      \end{matrix}
    \ \right\}.
  \end{align*}
  In the last two cases, terms are interpreted in an empty environment. We, then, omit the empty sequence associated with the empty environment, {\em e.g.}, $\alpha\cons \beta\cons\alpha$ stands for $((),\alpha\cons \beta\cons\alpha)$.\\
  We can verify that extensionality holds, indeed $\llb \underline{1}\rrb_{D}=\llb\I\rrb_{D}$. To prove it we use $\ext{D}$ as the witness function for both existential.
\end{example}

\begin{prop}\label{prop:eqIT}
  Let $M$ be a term of \Lamb and $D$ a filter model, the following statements are equivalent representations of the interpretation of $M$ in $D$:
  \begin{itemize}
  \item $(\vec \alpha,\beta)\in\llb M\rrb_D^{\vec x}$ for the interpretation defined in Figure~\ref{fig:intLam},
  \item the type judgment $\vec x:\alpha\vdash M:\beta$ is derivable by the rules of Figure~\ref{fig:tyLam1}.
  \end{itemize}
\end{prop}
\longv{
\begin{proof}
By structural induction on the grammar of $\Lamb$.
\end{proof}
}

\begin{definition}\label{def:sens}
  A DEFiM $D$ is sensible for the $\lambda$-calculus when $\llb M\rrb^{\vec x}=\emptyset$ iff $M\Ua$.
\end{definition}

\longv{
\begin{example}\label{ex:coind} \todo{To define uniquely as a type system...}
  Not every filter model can be obtained as the extensional completion of a simpler partial filter model. Using the correspondence of Proposition~\ref{prop:eqIT}, we can also use intersection type systems to define complex models. For example, the (positive) coinductive intersection types form a filter model of interest:

  Coinductive intersection types are generated by the following grammar, which add the coinductive pattern $\nu X.\alpha$ to the usual intersection types. Notice that we use syntactic $\wedge$, $\rta$ and $\epsilon$ temporarily to represent what will become the semantic ones in the model (where $X$ is a variable from a denumerable set): 
  \begin{align*}
    (D) \quad\quad  \alpha,\beta\ &:=\quad X \mid \alpha\wedge\beta \mid \omega \mid \nu X.(\alpha\cons\beta)
  \end{align*}
  this grammar is quotiented by the equations of filter model (Def.~\ref{def:filterModels}) modulo the coinduction:
  \begin{align*}
    \alpha\wedge\alpha &= \alpha &
    \alpha\wedge\beta &= \beta\wedge\alpha &
    \alpha\wedge(\beta\wedge\gamma) &= (\alpha\wedge\beta)\wedge\gamma &
    \alpha\wedge\omega &= \alpha &
    \nu X.(\alpha\cons\beta) = \nu Y.(\alpha[Y/X]\cons\beta[Y/X])
  \end{align*}
  and $\nu X.(\gamma\cons\delta)\ge_D\bigwedge_i\nu Y.(\alpha_i\cons\beta_i)$ whenever
  $$ \delta[\nu X.(\gamma\cons\delta)/X]\ge_D\bigwedge_{\{i\mid\gamma[\nu X.(\gamma\cons\delta)/X]\le\alpha_i[\nu Y.(\alpha_i\cons\beta_i)]/Y\}}\beta_i[\nu Y.(\alpha_i\cons\beta_i)/Y],$$
  
  For the sake of extensionality and sensibility, it is usual to restrict ourselves to close types and positive coinductive calls, which are the types $\alpha\in D_F$ such that $;\Vdash \alpha$ is provable in the system:
  \begin{center}
    \AxiomC{$\Delta;\Gamma, X\Vdash \alpha$}
    \AxiomC{$\Gamma, X ; \Delta\Vdash \beta$}
    \BinaryInfC{$\Gamma;\Delta\Vdash \nu X. \alpha\cons\beta$}
    \DisplayProof 
    \hspace{0.em}
    \AxiomC{$\phantom{G}$}
    \UnaryInfC{$\Gamma;\Delta\Vdash \omega$}
    \DisplayProof 
    \AxiomC{$\phantom G$}
    \hspace{0.em}
    \UnaryInfC{$\Gamma, X;\Delta\Vdash X$}
    \DisplayProof 
    \hspace{0.em}
    \AxiomC{$\Gamma;\Delta\Vdash \alpha$}
    \AxiomC{$\Gamma;\Delta\Vdash \beta$}
    \BinaryInfC{$\Gamma;\Delta\Vdash \alpha\wedge\beta$}
    \DisplayProof 
    \hspace{0.em}
    \AxiomC{$\Delta;\Gamma\Vdash \alpha$}
    \AxiomC{$\Gamma;\Delta\Vdash \beta$}
    \BinaryInfC{$\Gamma;\Delta\Vdash \alpha\cons\beta$}
    \DisplayProof 
  \end{center}
  This definition is correct because these rules distributes with the equations of filter models, excepts for the second which can be resolved trivially.
  This system can be shown distributive and extensional with:
  \begin{align*}
    \ext{D}(\alpha\wedge\beta) &:= \ext{D}(\alpha)\cup\ext{D}(\beta) &
    \ext{D}(\nu X.(\alpha\cons\beta)) &:= \{\alpha[(\nu X.(\alpha\cons\beta))/X],\beta[(\nu X.(\alpha\cons\beta))/X]\}
  \end{align*}
\end{example}
}

\begin{definition}[Sensibility]\label{def:sens}
  A filter model $D$ is sensible for the untyped $\lambda$-calculus if diverging terms corresponds exactly to those of empty interpretation:
    $$ M\Da \quad\Lra\quad \llb M\rrb^{\vec x} \neq \emptyset\ .$$
\end{definition}

\longv{
  Hereafter, $D$ denotes a fixed DEFiM.
}

  \subsection{B\"ohm Approximants} 
  \label{ssec:BT}
  
The B\"ohm approximants (or finite B\"ohm trees) are the normal forms of a $\lambda$-calculus extended with a constant\footnote{In other context, the constant $\Omega$ has been replaced by $\bot$.} $\Omega$ and an additional reduction $\rta_\Omega$.
\longv{

}
A \emph{$\lambda_\Omega$-term} $M$ is a $\lambda$-term possibly containing occurrences of the constant $\Omega$.
The set $\Lambda_\Omega$ of all \emph{$\lambda_\Omega$-terms} is generated by the grammar:
$$
	\Lambda_\Omega :\quad M,N\ ::=\ x\ \mid\ \lambda x.M\ \mid\ MN \ \mid\ \Omega
$$
\longv{
Similarly a \emph{(single hole) $\lambda_\Omega$-context }is a (single hole) context $C\llc\rrc-$ possibly containing occurrences of $\Omega$.
}
The \emph{$\Omega$-reduction} $\rta_\Omega$ is defined as the $\lambda_\Omega$-contextual closure of the rules:
$$
(\Omega)\qquad\qquad\qquad	\lambda x.\Omega\to \Omega\qquad\qquad\qquad \Omega\: M \to \Omega
$$
The $\beta$-reduction is extended to $\lambda$-terms in the obvious way. The interpretation of $\lambda_\Omega$-terms is the immediate extension of the interpretation of terms (Fig.\ref{fig:intLam}) plus the minimal interpretation given to the bottom: $\llb \Omega\rrb^{\vec x}:= \{(\vec\alpha,\omega)\mid \forall \vec \alpha\}$.
We write $\BA$ for the set of $\lambda$-terms in $\beta\Omega$-normal forms whose elements are denoted by $s,t,u,\dots$

The following characterization of $\beta\Omega$-normal forms is well known.

\begin{lemma}
Let $M\in\lambda$. We have $M\in\BA$ if and only if either $M=\Omega$ or $M$ has shape $\lambda x_1\dots x_n.x_iM_1\cdots M_k$ (for some $n,k\ge 0$) and each $M_i$ is $\beta\Omega$-normal.
\end{lemma}

The set of all B\"ohm approximants of $M$ can be obtained by calculating the direct approximants of all $\lambda$-terms $\beta$-convertible with $M$.
Only then will we fully describe the property of approximability for a filter model.

\begin{definition}\label{def:omegastuff} Let $M\in\lambda$. 
\begin{enumerate}
\item\label{def:omegastuff1}
	The \emph{direct approximant of $M$}, written $\dapprox{M}$, is the $\lambda$-term defined as:
	\begin{itemize}
	\item $\dapprox{M} := \Omega$ if $M = \lambda x_1\dots x_k.(\lambda y.M')NM_1\cdots M_k$,
	\item $\dapprox{M} := \lambda x_1\dots x_n.x_i\dapprox{M_1}\cdots \dapprox{M_k}$ if $M = \lambda x_1\dots x_n.x_iM_1\cdots M_k$,
	\end{itemize}
\item\label{def:omegastuff2} 
	The \emph{set of finite approximants of $M$} is defined by\longv{:
	$$
			\B{M} := \big\{\dapprox{M'} \mid M \rta^*_h M'\big\}\,.
	$$
        }\shortv{ $\B{M} := \big\{\dapprox{M'} \mid M \rta^*_h M'\big\}.$}
\end{enumerate}
\end{definition}

\begin{definition}\label{def:approx}
  A filter model is approximable iff the interpretation of any term $M\in\Lambda$ is the sup of its approximants:
  \longv{
  $$ \llb M \rrb^{\vec x} \quad = \quad \bigcup_{N\in \B{M}} \llb N\rrb^{\vec x}.$$
  }
  \shortv{ $\llb M \rrb^{\vec x}  = \bigcup \llb N\rrb^{\vec x}$ over all $N\in \B{M}$. }
 
\end{definition}

\section{$\lambda$-calculi with D-tests}
  \label{sec:D-tests}
%

\subsection{Syntax}\label{sec:TestSynatax}
The original idea of using {\em tests} to recover full abstraction (via a theorem of definability) is due to Bucciarelli {\em et al.} \cite{BCEM11}. In~\cite{These,B.H*OpeAspect}, the author caried a precise study of variants of Bucciarelli {\em et al.}'s calculus adapted to Krivin's models. Here we extend a bit his definition to get all DEFiMs.

Directly dependent on a given DEFiM $D$, the \Lcalcul with $D$-tests $\Lam{D}$ is, to some extent, an internal calculus for $D$. In fact, we will see that, for $D$ to be fully abstract for $\Lam{D}$, it is sufficient to be sensible (Th.~\ref{th:FAwT}). Notice that in the notation $\Lam{D}$, $\tau$ stands for tests and $D$ if the considered DEFiM.

The idea is to introduce tests as a new kind in the syntax. Tests \newdefinv{tests} $Q\in\newsym{\protect\Test{D}}$ are sort of co-terms,\longv{\footnote{We will see in Remark~\ref{rk:polarisedTests} that in a polarized context, the behavior of test does not correspond to co-term (or stack), but to commands (or processes), {\em i.e.}, to interactions between usual terms and fictive co-terms extracted from the semantics.}} in the sens that their interpretations $\llb Q\rrb^{x_1...x_n}\in (D^n\Rta \{*\})$ are maps from the context to the trivial model, which is a singleton $\{*\}$ where $*$ represents the convergence of the evaluation, seen as a success.

The interaction between terms and tests is carried out by two groups of syntactical constructors, each indexed by the elements $\alpha\in D$, and with the following kinds:\vspace{-0.2em}
  $$\tau_\alpha:\Lam{D}\rta\Test{D}\quad \text{ and }\quad \bar\tau_\alpha:\Test{D}\rta\Lam{D}.\vspace{-0.2em}$$

The first operation, \newsym{\protect\tau_\alpha}, will verify that its argument $M\in \Lam{D}$ has the point $\alpha$ in its interpretation. Intuitively, this is performed by recursively unfolding the B\"ohm tree of $M$ and succeeding ({\em i.e.}, converging) when $\alpha$ is in the interpretation of the finite unfolded B\"ohm tree. If $\alpha\not\in \llb M\rrb$, the test~$\tau_\alpha(M)$ will either diverge or refute (raising a $\0$ considered as an error). Concretely, it is an infinite application that feeds its argument with empty $\bar\tau$ operators.

The second operator, \newsym{\protect\bar\tau_\alpha}, simply constructs a term of interpretation $\da\alpha$ if its argument succeeds and diverges otherwise. Concretely, it is an infinite abstraction that runs its test argument, but also tests each of its applicants using $\tau$ operators.

In addition to these operators, we use \newdefprem{sums}{of terms} \newdefinvinv{sums}{of tests} and \newdefprem{products}{of tests} as ways to introduce may (for the addition) \newdefinvinv{may non-determinism}{in \Lam{D}} and must (for the multiplication) non-determinism\newdefinvinv{must non-determinism}{in \Lam{D}}; in the spirit of the $\lambda\+||$-calculus \cite{DLP98}. Indeed, these two forms of non-determinism are necessary to explore the branching of B\"ohm trees.

The idea of these two operators is to use the parametricity of our terms toward their intersection types. The term $\bar\tau_\alpha(\epsilon)$ (further on denoted by $\bareps_\alpha$), that transfers the always succeeding test $\epsilon$ into~a term of interpretation $\da\alpha$, constitutes the canonical term of type $\alpha$; its behavior is \longv{exactly }the common behavior of every term of type $\alpha$. Symmetrically, the test $\tau_\alpha(M)$ verifes whether $M$ behaves like a term of type $\alpha$.

\begin{definition}\label{def:LtauD-calculus}
  The \newdefpremsec{\Lcalcul}{with D-tests}, for short \newsym{\protect\Lam{D}}, is given by the following grammar:
  \begin{center}
    \begin{tabular}{l l r @{\ ::=\ } l  l}
      (term) & 
      $\Lam{D}$ & 
      $M,N$ \hspace{1em} & 
      \hspace{1em} $x\quad |\quad \lambda x.M\quad |\quad M\ N\quad |\quad \newsym{\protect\sum_{i\le n}\protect\bar\tau_{\alpha_i}(Q_i)}$ & 
      \hspace{1em} $,\forall (\alpha_i)_i\in (D-\omega)^n,n\ge 0$ 
      \vspace{0.5em}\\
      (test) & 
      $\Test{D}$ & 
      $P,Q$ \hspace{1em} & 
      \hspace{1em} $\newsym{\protect\sum_{i\le n}P_i}\quad |\quad \newsym{\protect\prod_{i\le n}P_i}\quad |\quad \protect\tau_{\alpha}(M)$ & 
      \hspace{1em} $,\forall \alpha\in (D-\omega),n\ge 0$ \\
    \end{tabular}
  \end{center}
  The empty sum is denoted by \newsym{\protect\0}, and the empty product by \newsym{\protect\eps}. Binary sums (resp. products) can be written with infix notation, {\em i.e.} $P\+Q$ (resp $P\pt Q$), but we will more than often use arbitrary finite sums $\Sigma_i P_i$ and products $\Pi_i P_i$. 
  
  Moreover, we use the notation $\newsym{\protect\bareps_\alpha}:=\bar\tau_\alpha(\eps)$ and $\newsym{\protect\bareps_a}:=\sum_{\alpha\in a}\bareps_\alpha$; which are terms.
  
  Sums and products are considered as multisets, in particular we suppose associativity, commutativity and neutrality with, respectively, $\0$ and $\eps$. 

  In the following, an \newdef{abstraction} can refer either to a $\lambda$-abstraction or to a sum of $\bar\tau$ operators. This notation is justified by the behavior of $\Sigma_i\bar\tau_{\alpha_i}(Q_i)$ that mimics an infinite abstraction.
  
  The operational semantics is given by three sets of rules in Figure~\ref{fig:OS}. The {\em main rules} of Figure~\ref{fig:ER} are the effective rewriting rules. The {\em distributive rules} of Figure~\ref{fig:DS} implement the distribution of the sum over the test-operators and the product. The small step semantics \newdefinvinv{reduction}{for \Lam{D}} \newsymprem{\protect\rta}{as reduction in \Lam{D}} is the free contextual closure \longv{({\em i.e.}, by the rules of Figure~\ref{fig:FCR})} of the rules of Figures~\ref{fig:ER} and ~\ref{fig:DS}. The {\em contextual rules} of Figure~\ref{fig:CR} implement the \newdefsecpreminv{reduction}{head}{for \Lam{D}} \newsymprem{\protect\rta_h}{for \Lam{D}} that is the specific contextual extension we are considering.
\end{definition}


\begin{figure} \caption{Operational semantics of the calculus with $D$-tests.\\ In rules $\bar\tau$ and $\tau$, notice that we use the notations $\tau_\omega(M) := \eps$ and $\bareps_\omega := \0$ in order to keep the rule simpler.}\label{fig:OS}
  \centering
    \begin{subfigure}[t]{\textwidth}
      \caption{Main rules}
      \label{fig:ER}
      \begin{minipage}{0.4\textwidth}
        \begin{center}
          \begin{tabular}{l c l}
            \vspace{1em}
            $(\beta)$\hspace{2em} $(\lambda x.M)\ N$ & $\!\rta\!$ & $M[N/x]$\\
            \vspace{1em}
            \newsymsec{reduction rule}{(\protect\bar\tau)}\hspace{2em}  $(\sum_i\bar\tau_{\alpha_i}(Q_i))\ N$ & $\!\rta\!$ & $\sum_i\sum_{(\beta,\gamma)\in \ext{D}(\alpha_i)}\bar\tau_{\gamma}(Q_i\ \pt\ \tau_\beta(N))$\\
            \vspace{1em}
            \newsymsec{reduction rule}{(\protect\tau)}\hspace{2em} $\tau_\alpha(\lambda x.M)$ & $\!\rta\!$ & $\prod_{(\beta,\gamma)\in \ext{D}(\alpha)}\tau_{\gamma}(M[\bareps_\beta/x])$\\
            \vspace{1em}
            \newsymsec{reduction rule}{(\protect\tau\protect\bar\tau)}\hspace{2em}  $\tau_\alpha(\sum_{i\in I}\bar\tau_{\beta_i}(Q_i))$ & $\!\rta\!$ & $\sum_{\{I'\subseteq I\mid \alpha\ge \bigwedge_{i\in I'}\beta_i\}}\prod_{i\in I'}Q_i$
          \end{tabular}
        \end{center}
      \end{minipage}
    \end{subfigure}

    \begin{subfigure}[t]{\textwidth}
      \caption{Distribution of the sum}  
      \label{fig:DS}
      \begin{minipage}{0.45\textwidth}
        \begin{center}
          \begin{tabular}{c r c l}
            \newsymsec{reduction rule}{(\protect\pt\+)}\hspace{4em} & $\Pi_{i\le n}\Sigma_{j\le k_i}Q_{i,j}$ &$\hspace{-0.5em}\rta\hspace{-0.5em}$ &$\Sigma_{j_1\le k_1,...,j_n\le k_n}\Pi_{i\le n}Q_{i,j_i}$
            \vspace{1em}\\
            \newsymsec{reduction rule}{(\protect\bar\tau+)}\hspace{4em} & $\bar\tau_\alpha(\Sigma_iQ_i)$ &$\hspace{-0.5em}\rta\hspace{-0.5em}$ &$\Sigma_i\bar\tau_\alpha(Q_i)$
          \end{tabular}
        \end{center}
      \end{minipage}
    \end{subfigure}
   
    \begin{subfigure}[t]{\textwidth}
      \caption{Contextual rules for the head reduction}
      \label{fig:CR}
      \begin{minipage}{\textwidth}
      \begin{minipage}{\textwidth}
        \begin{center}
          \AxiomC{$M\rta_h M'$}
          \RightLabel{\newsymsec{reduction rule}{(h\dash c\lambda)}}
          \UnaryInfC{$\lambda x.M\rta_h \lambda x.M'$}
          \DisplayProof \nolinebreak \hskip 50pt 
          \AxiomC{$M\rta_h M'$}
          \AxiomC{$M$ is an application}
          \RightLabel{\newsymsec{reduction rule}{(h\dash c\protect\at)}}
          \BinaryInfC{$M\ N\rta_h M'\ N$}
          \DisplayProof\\[0.7em]
          \AxiomC{$M\rta_h M'$}
          \AxiomC{$M$ is an application}
          \RightLabel{\newsymsec{reduction rule}{(h\dash c\protect\tau)}}
          \BinaryInfC{$\tau_\alpha(M)\rta_h \tau_\alpha(M')$}
          \DisplayProof\nolinebreak \hskip 40pt
          \AxiomC{$Q\rta_h Q'$}
          \AxiomC{\hspace{-1em}$Q$ is not a sum}
          \RightLabel{\newsymsec{reduction rule}{(h\dash c\protect\bar\tau)}}
          \BinaryInfC{$\bar\tau_\alpha(Q)\rta_h \bar\tau_\alpha(Q')$}
          \DisplayProof\\[0.7em]
          \AxiomC{$M\rta_h M'$}
          \RightLabel{\newsymsec{reduction rule}{(h\dash cs)}}
          \UnaryInfC{$M+N\rta_h M'+N$}
          \DisplayProof\nolinebreak \hskip 18pt
          \AxiomC{$Q\rta_h Q'$}
          \RightLabel{\newsymsec{reduction rule}{(h\dash c\protect+\protect\!)}}
          \UnaryInfC{$Q+P\rta_h Q'+P$}
          \DisplayProof\nolinebreak \hskip 18pt
          \AxiomC{$Q\rta_h Q'$}
          \AxiomC{\hspace{-1em}$Q$ is not a sum}
          \RightLabel{\newsymsec{reduction rule}{(h\dash c\protect\pt)}}
          \BinaryInfC{$Q\pt P\rta_h Q'\pt P$}
          \DisplayProof
        \end{center}
      \end{minipage}
      \end{minipage}
    \end{subfigure}

    \longv{
    \begin{subfigure}[t]{\textwidth}
      \caption{Contextual rules for the full reduction}
      \label{fig:FCR}
      \begin{minipage}{\textwidth}
      \begin{minipage}{\textwidth}
        \begin{center}
          \AxiomC{$M\rta M'$}
          \RightLabel{\newsymsec{reduction rule}{(c\lambda)}}
          \UnaryInfC{$\lambda x.M\rta \lambda x.M'$}
          \DisplayProof\nolinebreak\hskip 30pt
          \AxiomC{$M\rta M'$}
          \RightLabel{\newsymsec{reduction rule}{(c\protect\at L)}}
          \UnaryInfC{$M\ N\rta M'\ N$}
          \DisplayProof\nolinebreak\hskip 30pt
          \AxiomC{$N\rta N'$}
          \RightLabel{\newsymsec{reduction rule}{(c\protect\at R)}}
          \UnaryInfC{$M\ N\rta M\ N'$}
          \DisplayProof\\
          \vspace{1.3em}
          \AxiomC{$M\rta M'$}
          \RightLabel{\newsymsec{reduction rule}{(c\protect\tau)}}
          \UnaryInfC{$\tau_\alpha(M)\rta \tau_\alpha(M')$}
          \DisplayProof\nolinebreak\hskip 40pt
          \AxiomC{$Q\rta Q'$}
          \RightLabel{\newsymsec{reduction rule}{(c\protect\bar\tau)}}
          \UnaryInfC{$\bar\tau_\alpha(Q)\rta \bar\tau_\alpha(Q')$}
          \DisplayProof\\
          \vspace{1.3em}
          \AxiomC{$M\rta M'$}
          \RightLabel{\newsymsec{reduction rule}{(cs)}}
          \UnaryInfC{$M+N\rta M'+N$}
          \DisplayProof\nolinebreak\hskip 20pt
          \AxiomC{$Q\rta Q'$}
          \RightLabel{\newsymsec{reduction rule}{(c+\protect\!)}}
          \UnaryInfC{$Q+P\rta Q'+P$}
          \DisplayProof\nolinebreak\hskip 20pt
          \AxiomC{$Q\rta Q'$}
          \RightLabel{\newsymsec{reduction rule}{(c\protect\pt)}}
          \UnaryInfC{$Q\pt P\rta Q'\pt P$}
          \DisplayProof
        \end{center}
      \end{minipage}
      \end{minipage}
    \end{subfigure}
    }
  \end{figure}


\begin{example}
  The operational behavior of $D$-tests depends on $D$. Recall the DEFiMs of Example~\ref{example:1}. 
  In the case of Park $P_\infty$:\vspace{-0.2em}
  \begin{align*}
    \underline{\tau_*(\lambda x}.xx) \quad  
      & \stackrel{\tau}{\rta}_h \quad \tau_*(\underline{\bareps_*\ \bareps_*}) \quad
        \stackrel{\bar\tau}{\rta}_h \quad \underline{\tau_*(\bar\tau_*}(\underline{\tau_*(\bareps_*}))) \quad
        \stackrel{\tau\bar\tau}{\rta}_h\stackrel{\tau\bar\tau}{\rta}_h \quad \eps.
  \end{align*}
  In the case of Scott's $D_\infty$ we have in $\Lamb_{\tau(D_\infty)}$:\vspace{-0.2em}
  \begin{align*}
    \tau_*(\underline{(\lambda xy.x\ y)\ \bareps_*}) \quad 
      & \stackrel{\beta}{\rta}_h\quad \underline{\tau_*(\lambda y}.\bareps_*\ y) \quad 
        \stackrel{\tau}{\rta}_h \quad \tau_*(\underline{\bareps_*\ \0})  \quad
       \stackrel{\bar\tau}{\rta}_h \quad \tau_*(\underline{\bareps_*}) \quad
        = \quad \underline{\tau_*(\bar\tau_*}(\eps)) \quad
        \stackrel{\tau\bar\tau}{\rta}_h \quad \eps, \\
    \tau_*(\underline{(\lambda xy.y\ x)\ \bareps_*}) \quad
      & \stackrel{\beta}{\rta}_h \quad \underline{\tau_*(\lambda y}.y\ \bareps_*) \quad
        \stackrel{\tau}{\rta}_h \quad \tau_*(\underline{\0}\ \bareps_*) \quad
        \stackrel{\bar\tau}{\rta}_h \quad \underline{\tau_*(\0)} \quad
        \stackrel{\tau\bar\tau}{\rta}_h \quad \0.
  \end{align*}
  In the case of Norm:\vspace{-0.2em}
  \begin{align*}
    \underline{\tau_p(\lambda x}.x) \quad
      & \stackrel{\tau}{\rta}_h \quad \underline{\tau_p(\bareps_q)} \quad 
        \stackrel{\tau\bar\tau}{\rta}_h \quad \eps,
        &&\text{and }
    &\underline{\tau_q(\lambda x}.x) \quad
      & \stackrel{\tau}{\rta}_h \quad \underline{\tau_q(\bareps_p)} \quad
        \stackrel{\tau\bar\tau}{\rta}_h \quad \0.
  \end{align*}
 In the case of $Z_ \infty$:\vspace{-0.2em}
  \begin{align*}
    \tau_{\underline{n+2}}(\bareps_{\underline{n}}\ M)\stackrel{\bar\tau}{\rta}_h\:&\tau_{\underline{n+2}}(\bareps_{n+1})\stackrel{\tau\bar\tau}{\rta}_h \boldsymbol{0}\ , &
    \tau_{\underline{n+2}}(\bareps_{\underline{n}}\ M\ N)\stackrel{\bar\tau}{\rta}_h{\!\!\!\!}^2\:&\tau_{\underline{n+2}}(\bareps_{\underline{n+2}})\stackrel{\tau\bar\tau}{\rta}_h \eps\ , &
    \tau_{\underline{n+2}}(\bareps_{\underline{n}}\ M\ N\ L)
       \longv{\stackrel{\bar\tau}{\rta}_h{\!\!\!\!}^3\:&\tau_{\underline{n+2}}(\bareps_{\underline{n+3}})\stackrel{\tau\bar\tau}{\rta}_h}
       \shortv{\rta_h^4}
       \boldsymbol{0}
  \end{align*} 
\end{example}


\longv{
\begin{remark}\label{rk:polarisedTests}
  In a polarized (or classical) framework with explicit co-terms (or stacks) as the framework presented in \cite{Mun09}, tests would correspond to commands (or processes), or, more exactly, to conjunctions and disjunctions of commands. Indeed, a test $\tau_\alpha(M)$ is nothing else than the command $\langle M\mid\pi_\alpha\rangle$ where $\pi_\alpha$ would be the canonical co-term of interpretation $\ua\alpha$, the same way that $\bareps_\alpha$ is the canonical term of interpretation $\da\alpha$. Similarly, the term $\bar\tau(Q)$ can be seen as the canonical term $\bareps_\alpha$ endowed with a parallel composition referring to the set of commands $Q$. To resume, we have:
  \begin{align*}
    \tau_\alpha(M)\ &\simeq\  \langle M \mid \ua\alpha \rangle   &   \langle \bar\tau_\alpha(Q)\mid\pi\rangle\ &\simeq\ \langle \da\alpha \mid \pi\rangle\pt Q
  \end{align*}
\end{remark}
}


\begin{definition}
A test is in {\em may-head-normal form} if it has the shape
$\Pi_i\tau_{\alpha_{i}}(x_{i}\ M^1_{i}\cdots\ M^n_{i}) + Q,$
with $i\ge 0$ and $M^k_{i}$ any term.
A term is in \emph{may-head-normal form}, either if it has one of the shape
$(\lambda x_1....x_n.y\ M_1 \cdots\ M_m)$,    
or $\lambda x_1...x_n. \bar\tau_{\alpha}(Q)+N$,
where $m,n\ge 0$, $\alpha\in (D-\omega)$, $M_i$ and $N$ any terms, and $Q$ any test in head-normal form without sums. 
Coherently with the head convergence in \Lcalcul, the convergence to a may-head-normal form will be denoted by \newsymprem{\protect\Da}{in \Lam{D}}\newsyminvinv{\protect\Da N}{in \Lam{D}} and the divergence by \newsymprem{\protect\Ua}{in \Lam{D}}.
\end{definition}

\begin{example}
  For any $n\in\Nat$, the term $\underline{n}\ (\lambda x. \bar\tau_\alpha(\tau_\alpha(x)\+\tau_\beta(x)))\ \bareps_\alpha$ may-head-converges.
\end{example}

Let us notice that this calculus enjoys the properties of confluence and standardization and a powerfull property stating tests-reductions can always be postponed until the very end. \cite{These,B.H*OpeAspect}\longv{\todo{Consider writing them formally}} 
\longv{
It also enjoys a very nice property stating that tests-reductions can always be postponed until the very end:

\begin{theorem}\label{th:BetaFirst}
  Let $D$ a DEFiM and $M,N\in\Lam{D}$.\\
  For any reduction $M\rta^* N$, there exists $M',N'\in\Lam{D}$ such that $M\rta_\beta^*M'$ with only $\beta$-reductions, $M'\rta^*_{\backslash\!\!\!\beta}$ with only tests reductions, and $N\rta^*N$.
  \longv{
  \begin{alignat*}3
    \ M & \ \rta^* & \ N \\
    {}_\beta\rotatebox[origin=c]{-90}{$\rta^{*}$}\ & \ \ \ \rotatebox[origin=c]{-90}{$\rightsquigarrow$} & \rotatebox[origin=c]{-90}{$\rta^*$}\ \\
    \ M' & \ \rta^*_{\backslash\!\!\!\beta} & \ N' 
  \end{alignat*}
  }
  In particular, $M$ is may-head converging iff there is a sequence of $\beta$-reductions $M\rta^*_\beta L$ with $L$ that is may-head converging without any $\beta$-reduction.
\end{theorem}
}

%
%

\longv{
\begin{definition}
Grammars of \newdef{term-context}s \Lamcont{D} and \newdef{test-context}s \Testc{D} are given in Figure~\ref{fig:grC}. 
\end{definition}

\begin{figure*}
  \caption{Grammar of the contexts in a calculus with $D$-tests \label{fig:grC}}
  \begin{tabular}{l l r @{\ ::=\ } l l}
    \!\!(term-context)\!\! &
    \newsym{\protect\Lamcont{D}} & 
    $C$  & 
    $x\quad |\quad \newsymprem{\protect\llc.\protect\rrc}{for \protect\Lam{D}}\quad |\quad C \ C' \quad |\quad \lambda x.C \quad |\quad \sum_{i\le n}\bar\tau_{\alpha_i}(K_i )$ & 
    $,\forall (\alpha_i)_i\in D^n,n\ge 0$ 
    \vspace{0.5em}\\
    \!\!(test-context)\!\!  & 
    \newsym{\protect\Testc{D}}  & 
    $K$  & 
    $\sum_{i\le n}K_i \quad |\quad \prod_{i\le n}K_i \quad |\quad \tau_{\alpha}(C )$ & 
    $,\forall \alpha\in D,n\ge 0$ \\
  \end{tabular}
\end{figure*}

\begin{definition}
The \newdefprem{observational preorder}{of \Lam{D}} \newsym{\protect\leobtau{D}} of \Lam{D} is defined by: \vspace{-0.5em}
 $$M\sqsubseteq_{\tau(D)}N\ \ \text{iff}\ \ (\forall K \mathrm{\in} \Testc{D},\ K\llc M\rrc\Da\ \text{ implies }\ K\llc N\rrc\Da).$$
We let \newsym{\protect\equivobtau{D}} denote the \newdefprem{observational equivalence}{of \Lam{D}}, {\em i.e.}, the equivalence induced by \leobtau{D}.
\end{definition}

\begin{remark}
  The observational preorder could have been defined using term-contexts rather than test-contexts, but this appears to be equivalent and test-contexts are easier to manipulate (because normal forms for tests are simpler).
\end{remark}
}

\subsection{Semantics}\label{subsubsec:TestSemantics}
The standard interpretation of $\Lamb$ into $D$ can be extended to~$\Lam{D}$ (Fig.~\ref{fig:intTests}).
\begin{figure*} \label{fig:TestInter}
  \begin{subfigure}[t]{\textwidth}
    \caption{Interpretation of $\Lamb$ \longv{(copy of Figure~\ref{fig:intLam})}\shortv{\label{fig:intLam}}}
    \vspace{-1em}
    \begin{center}
      $\llb x_i \rrb_D^{\vec x}  =  \{(\vec \alpha,\beta)\ |\ \beta\ge\alpha_i\}$
      \hspace{5em}
      $ \llb \lambda y.M \rrb_D^{\vec x}  =  \{(\vec \alpha,\bigwedge_i(\beta_i\cons\gamma_i))\ |\ \forall i,\ (\vec \alpha\beta_i,\gamma_i)\in\llb M\rrb_D^{\vec xy}\}$
      \vspace{0.3em}\\
      $\!\llb M\ N \rrb_D^{\vec x}  =  \{(\vec \alpha,\bigwedge_i\beta_i)\ |\ \exists \vec\gamma_i,(\vec \alpha,\bigwedge_i(\gamma_i\cons\beta_i))\in\llb M\rrb_D^{\vec x}\ \wedge (\vec \alpha,\bigwedge_i\gamma_i)\in\llb N\rrb_D^{\vec x}\}$
    \end{center}
  \end{subfigure}
  \begin{subfigure}[t]{\textwidth}
    \caption{Interpretation of tests extensions \label{fig:intTests}}
    \vspace{-1em}
    \begin{center}
      $\llb \Sigma_{i\in J}\bar\tau_{\alpha_i}(Q_i) \rrb_D^{\vec x}  =  \{(\vec \beta,\gamma)\ |\ \exists I\subseteq J, \vec \beta\in \bigcap_{i\in I}\llb Q_i\rrb_D^{\vec x}\ \wedge\ \gamma\ge_{D}\bigwedge_{i\in I}\alpha_i\}$
      \\[0.4em]
      $\llb \0 \rrb^{\vec x}_D  =  \{(\vec\alpha,\omega)\}$
      \hspace{1em}
      $\llb \tau_{\alpha}(M) \rrb_D^{\vec x}  =  \{\vec \beta\ |\ (\vec \beta,\alpha)\in\llb M\rrb_D^{\vec x}\}$
      \\[0.4em]
      $\llb \Pi_{i\le k}Q_i \rrb_D^{\vec x}  =  \bigcap_{i\le k} \llb Q_i\rrb_D^{\vec x}$
      \hspace{1em}
      $\llb \eps \rrb^{\vec x}_D  =  D^{\vec x}$
      \hspace{1em}
      $\llb \Sigma_{i\le k}Q_i \rrb_D^{\vec x}  =  \bigcup_{i\le k}\llb Q_i\rrb_D^{\vec x}$
      \hspace{1em}
      $\llb \0 \rrb^{\vec x}_D  =  \emptyset$
    \end{center}
  \end{subfigure}
  \begin{subfigure}[t]{\textwidth}
    \caption{Intersection types for the $\lambda$-calculus in $D$ \label{fig:tyLam} \longv{(copy of figure~\ref{fig:tyLam1})}\shortv{\label{fig:tyLam1}}}
    \vspace{-1em}
      \begin{center}
        \AxiomC{$\phantom{M}$}
        \UnaryInfC{$x:\alpha\vdash x:\alpha$}
        \DisplayProof\hskip 50pt
        \AxiomC{$\Gamma\vdash M:\alpha$}
        \UnaryInfC{$\Gamma,x:\beta\vdash M:\alpha$}
        \DisplayProof\hskip 50pt
        \AxiomC{$\Gamma\vdash M:\beta$}
        \AxiomC{$\alpha\ge\beta$}
        \BinaryInfC{$\Gamma\vdash M:\alpha$}
        \DisplayProof\\[0.5em]
        \AxiomC{$\Gamma,x:\alpha\vdash M:\beta$}
        \UnaryInfC{$\Gamma\vdash \lambda x.M:\alpha\cons\beta$}
        \DisplayProof\hskip 30pt
        \AxiomC{$\Gamma\vdash M:\alpha\cons\beta$}
        \AxiomC{$\Gamma\vdash N:\alpha$}
        \BinaryInfC{$\Gamma\vdash M\ N:\beta$}
        \DisplayProof\hskip 30pt
        \AxiomC{$\Gamma\vdash M:\alpha$}
        \AxiomC{$\Gamma\vdash M:\beta$}
        \BinaryInfC{$\Gamma\vdash M:\alpha\wedge\beta$}
        \DisplayProof
      \end{center}
  \end{subfigure}
  \begin{subfigure}[t]{\textwidth}
    \caption{Intersection types for the $D$-tests extension in $D$ \label{fig:tyTests}}
    \vspace{-1em}
      \begin{center}
        \AxiomC{$\Gamma\vdash M:\alpha$}
        \UnaryInfC{$\Gamma\vdash \tau_\alpha(M)$}
        \DisplayProof\hskip 40pt
        \AxiomC{$\Gamma\vdash Q_j $}
        \UnaryInfC{$\Gamma\vdash \sum_{i\in I}\bar\tau_{\alpha_i}(Q_i):\alpha_j$}
        \DisplayProof\hskip 40pt
        \AxiomC{$\Gamma\vdash Q_j $}
        \UnaryInfC{$\Gamma\vdash \sum_{i\in I}Q_i$}
        \DisplayProof\hskip 40pt
        \AxiomC{$\forall i\in I,\ \Gamma\vdash Q_i$}
        \UnaryInfC{$\Gamma\vdash \prod_{i\in I}Q_i$}
        \DisplayProof
      \end{center} 
  \end{subfigure}
  \caption{Direct interpretation and intersection type system computing the interpretation in $D$ \label{fig:typing}}
\end{figure*}

\begin{definition}
  A term $M$ with $n$ free variables is {\em interpreted} as a morphism (Scott-continuous function) from $D^n$ to $D$ and a test $Q$ with $n$ free variables as a morphism from $D^n$ to the dualizing object $\{*\}$. Concretely, we use the Cartesian closeness to define $\llb M\rrb_D^{\vec x}$ as a downward-close sets of $(D^{op})^{\vec x}\times D$ and $\llb Q\rrb_D^{\vec x}$ as a downward-close subsets of $(D^{op})^{\vec x}$.

  This interpretation is given in Figures~\ref{fig:intLam} and~\ref{fig:intTests} by structural induction. 
\end{definition}

\begin{prop}\label{prop:testModel}
  Any DEFiM $D$ is a model for its own test extension (the \Lcalcul with $D$-tests), in the sens that the interpretation is contextual and invariant under reduction.
\end{prop}
\begin{proof}
The invariance under $\beta$-reduction  is obtained, as usual, by the Cartesian closedness of of the considered category of domains.\longv{\todo{introduce ScottL at some point?}} The other rules are easy to check directly.
\end{proof}

The idea of intersection types can be generalized to to tests as shown in Figure~\ref{fig:tyTests}. Notice that tests have no type: a test does not carry any behavior, and under a specific environment it can only be succeeding (and typable) or diverging (untypable). 

\begin{theorem}[Intersection types]
  Let D be a DEFiM and $M$ a term of $\Lam{D}$  (resp. $Q$ a test of $\Test{D}$), the following statements are equivalent:
  \begin{itemize}
  \item $(\vec \alpha,\beta)\in\llb M\rrb_D^{\vec x}$ (resp. $\vec \alpha\in \llb Q\rrb_D^{\vec x}$) in the interpretation of Figures~\ref{fig:intLam} and~\ref{fig:intTests},
  \item the type judgment $\vec x:\vec \alpha\vdash M:\beta$ (resp. $\vec x:\vec \alpha\vdash Q$) is derivable by the rules of Figures~\ref{fig:tyLam} and~\ref{fig:tyTests}.
  \end{itemize}
\end{theorem}
\begin{proof}
By structural induction on the grammar of $\Lam{D}$.
\end{proof}


\longv{\todo{ref to a forthcoming article?}}

Notice that the interpretation allows the following trivial lemma:
\begin{lemma}\label{lemma:test}
  If $D$ is sensible for $\Lamb_{\tau(D)}$ then:
  \begin{align*}
    (\vec \alpha\beta,\gamma)\in\llb M\rrb^{\vec x y} \quad &\Lra\quad (\vec \alpha,\gamma)\in\llb M[\bareps_\beta/y]\rrb^{\vec x}, &
    (\vec \alpha, \gamma)\in\llb M\rrb^{\vec x} \quad &\Lra\quad \vec \alpha\in\llb \tau_\gamma(M)\rrb^{\vec x}.
  \end{align*}
\end{lemma}

\subsubsection{Full abstraction and sensibility for tests}\label{subsubsec:FAwT}
The main interest of the full abstraction with tests is to be fully abstract as soon as it is sensible (Theorem~\ref{th:FAwT}). The sensibility is a very commune property saying that diverging terms are collapsed together and separated from non-diverging terms. In other worlds, such a model is able to give meaning to terminating terms and those only. The full abstraction, however, is a much stronger property stating that the equality in the model corresponds exactly to the observational equality (for the head-convergence). Collapsing those two properties gives the real meaning of tests: they are syntactical representation of ``reasonable'' domains. Where ``reasonable'' means extensional and (as we will see later on) approximable domains.

\begin{definition}\label{def:sensibilityWTests}
  A DEFiM $D$ is {\em sensible for \Lam{D}} whenever diverging terms (resp. tests) correspond exactly to the terms (resp. tests) having empty interpretation, {\em i.e.}, for all $M\in\Lam{D}$ and $Q\in \Test{D}$:
  \begin{align*}
    M\Ua &\quad \Lra \quad \llb M\rrb_D^{\vec x}=\{(\vec\alpha,\omega) \mid \forall \vec \alpha\}
    & Q\Ua &\quad \Lra \quad \llb Q\rrb_D^{\vec x}=\emptyset
  \end{align*}
\end{definition}

The following is an immediate theorems (the second is an application of the first):

\begin{theorem}[Definability]\label{th:caracFAwT}
  If $D$ is sensible for $\Lamb_{\tau(D)}$ then:
  $$ (\vec \alpha,\beta)\in\llb M\rrb^{\vec x}\ \ \Lra\ \ \tau_\beta(M[(\bareps_{\alpha_i}/x_i)_{i\le n}]) \Da.$$
\end{theorem}
\begin{proof}
  If $(\vec \alpha,\beta)\in\llb M\rrb^{\vec x}$ then $\llb\tau_\beta(M[(\bareps_{\alpha_i}/x_i)_{i\le n}])\rrb$ is not empty by Lemma~\ref{lemma:test}, thus it converges by sensibility. Conversely, if $\tau_\beta(M[(\bareps_{\alpha_i}/x_i)_{i\le n}])\Da$, since it has no variable, its interpretation is either empty or $\{()\}$, it has to be the second by sensibility, which means $(\vec \alpha,\beta)\in\llb M\rrb^{\vec x}$ (by Lemma~\ref{lemma:test}).
\end{proof}

\begin{theorem}[full abstraction]\label{th:FAwT}
  For any DEFiM $D$, if $D$ is sensible for $\Lam{D}$, then $D$ is inequationaly fully abstract for the observational preorder of \Lam D: \newdefinvinv{inequational full abstraction}{of $D$ for \Lam{D}}
  $$ \llb M\rrb\subseteq \llb N\rrb\quad \Lra\quad \forall C \in \Testc{D}, C\llc M\rrc\Da\Rta C\llc N\rrc\Da.$$
\end{theorem}
\longv{
  \begin{proof}
    Let $\llb M\rrb\subseteq \llb N\rrb$ and $C\llc M\rrc\Da$. Then by sensibility we have that $\llb C\llc M\rrc\rrb$ is non-empty. Moreover, by Proposition~\ref{prop:testModel} we have that $\llb C\llc M\rrc\rrb\subseteq \llb C\llc N\rrc\rrb$. Thus $\llb C\llc N\rrc\rrb$ is non-empty and by sensibility, $C\llc N\rrc\Da$.\\
    Conversely, suppose that for all context $C \in \Testc{D}, C\llc M\rrc\Da\Rta C\llc N\rrc\Da$ and let $(\vec \alpha,\beta)\in\llb M\rrb^{\vec x}$:\\
    Then by Theorem~\ref{th:caracFAwT}, $\tau_\beta(M[(\bareps_{\alpha_i}/x_i)_{i\le n}]) \Da$ where $n$ is the length of $\vec \alpha$. Thus, after stating the context $C=\tau_\beta((\lambda x_1...x_n.\llc.\rrc)\ \bareps_{\alpha_1}\cdots\bareps_{\alpha_n})$, we have  $C\llc M\rrc\rta_h^n\tau_\beta(M[(\bareps_{\alpha_i}/x_i)_{i\le n}])\Da$ which implies that $C\llc N\rrc\Da$. However, there is no choice 
    for the $n$ first head reductions of $C\llc N\rrc$, those are forced to be $C\llc N\rrc\rta^n_h\tau_\beta(N[(\bareps_{\alpha_i}/x_i)_{i\le n}])$ so that this term is also head-converging. Then by applying the reverse implication of Theorem~\ref{th:caracFAwT} we conclude $(\vec \alpha,\beta)\in \llb N\rrb^{\vec x}$.
  \end{proof}
}


\shortv{\newpage}
\section{Collapsing Sensibility and Approximability for Tests}
  \label{ssec:BTandTests}
%

Once we have said that sensibility and full abstraction are equivalent properties for test, it should not surprise the reader to learn that approximability is also equivalent to those properties. Indeed, approximability usually corresponds to the adequation of the B\"ohm-tree's equality, which is a property between sensibility and full abstraction. However, the situation is a bit mere subtle: if the properties of sensibility and full abstraction for $\Lam{D}$ strongly refer to tests mechanisms, the property of approximability is defined independently from tests. This really means that $D$-tests will behave well exactly whenever $D$ is approximable.

First we extend the languages of approximants with tests (or rather the language of tests with approximants):

\begin{theorem}
  The properties of $\Lam D$ (such as confluence, standardization, or Theorems~\ref{th:caracFAwT} and~\ref{th:FAwT}) are still true when adding to the calculus with $D$-test the term $\Omega$ and the rules:
  $$
  \lambda x.\Omega\ \to\  \Omega  \qquad\qquad\qquad  \Omega\ M\  \to\  \Omega   \qquad\qquad\qquad \tau_\alpha(\Omega)\ \to\  \0.
  $$  
\end{theorem}
\begin{proof}
  The term $\Omega$ behave similarly to the empty sum of terms $\0$. The only difference is the rule $\lambda x.\Omega \to \Omega$ which is an $\eta$-reduction and is fine due to $D$ being extensional.
\end{proof}

We can now use the approximants of Definition~\ref{def:omegastuff} together with tests:

\begin{lemma}\label{lemma:app=sens}
  For any DEFiM $D$, any sequence $\vec\alpha\in D^{\vec x}$, any $\beta\in D{-}\{\omega\}$ and any $M\in\Lambda$ (with free variables $\vec x$), the following are equivalent:
  \begin{itemize}
  \item the test $\tau_\beta(M[\overline{\bareps_\alpha}/\vec x])$ is may-head converging without $\beta$-reduction,
  \item the test with approximants $\tau_\beta(\dapprox{M}[\overline{\bareps_\alpha}/\vec x])$ is may-head converging,
  \item $(\vec\alpha,\beta)\in \llb \dapprox{M}\rrb^{\vec x}$.
  \end{itemize}
\end{lemma}
\longv{\todo{should be clearer...}}
\begin{proof}
  Considering that $\Omega$ is a notation for $\0$, the second and third points are equivalent by Theorem~\ref{th:caracFAwT}. The equivalence between the two first points is obtained by induction on $\dapprox M$:
  \begin{itemize}
  \item Immediate when $\dapprox M = M = x_i$.
  \item When $\dapprox M = \lambda y.\dapprox N$ for $M = \lambda y.N$, we can use the induction hypothesis on $N$.
  \item When $\dapprox M= \Omega$, this means that $\tau_\beta(M[\overline{\bareps_\alpha}/\vec x])\rta^* \tau_\beta'((\lambda y.M')\ M_1\cdots M_n)$ cannot converges without performing a $\beta$-reduction.
  \item Otherwise, $\dapprox M= x_i \dapprox {N_1}\cdots\dapprox{N_n}$ with $M = x_i N_1\cdots N_n$ thus the terms $\tau_\beta(M[\overline{\bareps_\alpha}/\vec x])$ and $\tau_\beta(\dapprox{M}[\overline{\bareps_\alpha}/\vec x])$
    can perform the same sequence of $\bar\tau$-reductions followed by a $\tau\bar\tau$-reduction which results in a sum and product combination of tests behaving the same way by induction hypothesis.\vspace{-1em}
  \end{itemize}
\end{proof}

This clearly shows that taking the approximants is an operation that distribute with the semantics. This is sufficient to get the approximation theorem whenever the extension with tests is sensible.

\begin{theorem}\label{th:app=sens}
  Any extensional filter model $D$, is approximable if and only if it is sensible for $D$-tests.
\end{theorem}
\begin{proof}
  Both implications are considered separately\shortv{ using Lemma~\ref{lemma:app=sens}}.
  \longv{
   \begin{itemize}
   \item If $D$ is sensible for $\Lam D$ then it is approximable:\\
     Let $\vec\alpha\in D^{\vec x}$, $\beta\in D{-}\omega$ and $M\in\Lambda$.
    \begin{itemize}
     \item If $(\vec\alpha,\beta)\in\llb\dapprox N\rrb_D^{\vec x}$ for some $M\rta^*N$, then $\tau_\alpha(N[\overline{\bareps_\alpha}/\vec x])\Da$ by Lemma~\ref{lemma:app=sens}, thus $\tau_\alpha(M[\overline{\bareps_\alpha}/\vec x])\Da$ and $(\vec\alpha,\beta)\in\llb M\rrb_D^{\vec x}$.
     \item If $(\vec\alpha,\beta)\in\llb M\rrb_D^{\vec x}$, then $\tau_\alpha(M[\overline{\bareps_\alpha}/\vec x])\Da$. By Theorem~\ref{th:BetaFirst}, $M\rta^*_\beta N$ with $\tau_\alpha(N[\overline{\bareps_\alpha}/\vec x])$ that may-head converges without $\beta$-reduction. Thus, $(\vec\alpha,\beta)\in \llb \dapprox{N}\rrb^{\vec x}$ by Lemma~\ref{lemma:app=sens}.
     \end{itemize}
   \item If $D$ is approximable then it is sensible for $\Lam D$:\\
   Let $\vec\alpha\in D^{\vec x}$, $\beta\in D{-}\omega$ and $M\in\Lambda$.
   \begin{itemize}
   \item If $\tau_\alpha(M[\overline{\bareps_\alpha}/\vec x])\Da$, then by Theorem~\ref{th:BetaFirst}, $M\rta^*_\beta N$ with $\tau_\alpha(N[\overline{\bareps_\alpha}/\vec x])$ that may-head converges without $\beta$-reduction. Thus, by Lemma~\ref{lemma:app=sens}, $(\vec\alpha,\beta)\in \llb \dapprox{N}\rrb^{\vec x}$, which is included in $\llb M\rrb^{\vec x}$ by approximability.
   \item If $(\vec\alpha,\beta)\in\llb M\rrb_D^{\vec x}$, then there is $M\rta^* N$ such that $(\vec\alpha,\beta)\in \llb \dapprox{N}\rrb^{\vec x}$. By Lemma~\ref{lemma:app=sens} we have $\tau_\beta(M[\overline{\bareps_\alpha}/\vec x])$ that is may-head converging so that $\tau_\alpha(M[\overline{\bareps_\alpha}/\vec x])\Da$.
   \end{itemize}\vspace{-1em}
 \end{itemize}
}
\end{proof}

\section{Sufficient Condition for the Sensibility of Tests}
  \label{ssec:realisability}
  \newcommand{\bigslant}[2]{{{#1}\!}\raisebox{-.2em}{${/}$}\!\raisebox{-.30em}{$#2$}}

So far we could not find a generic and uniform proof of the approximation property in the literature\longv{ for standard filter models}.\footnote{Save Chapter 17.3 of the book of Barendregt, Dekkers ans Statman~\cite{BarDekSta} where this proof is done in parallel for several models of different classes, missing uniformity.} Hence, we give a sufficient condition (Def.~\ref{def:wpos}) for a filter model $D$ to be approximable (Th.~\ref{th:sensibility}). We use this condition for stating the approximability of models from Example~\ref{example:1} (save for $P_\infty$)\longv{ and Example~\ref{ex:coind}}.

Here, we make a strong use of the equivalence between approximability and sensibility with tests (Th.~\ref{th:app=sens}) proven in the previous chapter. Indeed, if approximability is also proved using Tait reducibility methods~\cite{Tait67}, the process is not as well understood as in the proofs of sensibility. By directly relying on the connection with tests, we can get the more refined analysis of the theorem of approximation that we have ever find. 

After our detailed analysis, we describe a sufficient, but not necessary, condition for the approximability. Generalizing the study of sensible models carried out by Berline~\cite{Ber00} and her students (Kerth~\cite{KerthPhD} in particular). In fact, we include (by far) all filter models proven sensible in the literature!

\subsection{Realizers}
\newcommand\Sat{\mathtt{Sat}}

\longv{
  \begin{definition}
    A \newdef{saturated set} $S\in\Sat_D$ is a set of term $S\subseteq\Lambda_{\tau(D)}$ that is close by backward reduction.\\
    Given two saturated sets $S,T$, we let $S\mapsto T$ denote the saturated set of terms $M$ such that $(M\ N)\in T$ whenever $N\in S$.\\
  \end{definition}
  
  \begin{definition}\label{def:realizerIInD} 
    A \newdef{realizer} of $D$ in $\Lambda$ is a function $\newsym{\Rea}$ from $D$ to saturated subsets of $\Lambda$ such that for all $\alpha,\beta\in D$, we have
    $$\Rea(\alpha\wedge\beta)=\Rea(\alpha)\cap \Rea(\beta)    \quad\quad\quad\quad\quad     \Rea(\alpha\cons\beta)=\Rea(\alpha)\mapsto \Rea(\beta):=\{M\mid\forall N\in\Rea(\alpha),\: (M\:N)\in\Rea(\beta)\} .$$
    Given any $D$-indexed sequence $S$ of saturated sets, a realizer $\Rea$ of $D$ in $\Lambda$ is a $S$-realizer if for all $\alpha$, $\Rea(\alpha)\in S_\alpha$.\\
    This definition trivially is extended for a partial DEFiM $J$ in place of $D$.\\
  \end{definition}
  
  \begin{definition}
    We use the notation:
    \begin{itemize}
    \item $\newsym{\protect\mathcal{N}_\Lambda^+} := \{M\in\Lambda\ |\ M\Da \}$,
    \item $\newsym{\protect\mathcal{N}_\Lambda^-} := \{x\ M_1\cdots M_k\ |\ x\in\Var, k\ge 0, M_1,...,M_k\in\Lambda \}$,
    \item for all $\alpha\in D{-}\omega$, $S_\Lambda^\alpha$ is the set of saturated subsets of $\mathcal{N}_\Lambda^+$ that contains $\mathcal{N}_\Lambda^-$,
    \item $S_\Lambda^\omega=\Lambda$
    \item $S_\Lambda^{D}=(S_\Lambda^\alpha)_{\alpha\in D}$.
    \end{itemize}
    For any partial DEFiM $J\subseteq D$, we write $S^{J}_\Lambda$ for the restriction 
    to $J$.
  \end{definition}
  
  \begin{lemma}\label{lemma:AdInt}
    Let $\Rea$ be a $S_\Lambda^{D}$-realizer in $D$.
    \begin{alignat*}4
      \text{if} &\quad (\vec a,\alpha)\in\llb M\rrb^{\vec x}\quad &\text{and} \quad (\forall i,L_i\in \Rea(a_i))\quad &\text{then} \quad M[\vec L/\vec x]\in \Rea(\alpha)
    \end{alignat*}
  \end{lemma}

  \begin{theorem}\label{th:sensibility}
    A DEFiM $D$ is sensible for $\Lambda$ iff it has a $S_\Lambda^{D}$-realizer of $D$ in $\Lambda$.
  \end{theorem}

  \begin{definition}
    $\tau(D)$-saturated sets, and \newdef{realizer} of $D$ in $\Lam D$ are defined similarly, excepts that the considered calculus is the calculus with tests.
  \end{definition}

}

Intuitively, a $S$-realizer is a proof that a certain property represented by $S$ is true for every typable term. This ``certain property'' is basically the commune property of elements of $S_\alpha$ (for $\alpha\neq\omega$). In our case, we are looking for sensibility, this gives us the sequence $S$ described by:

\begin{definition}
  We write, for all $\alpha\in D{-}\omega$:
  \begin{itemize}
  \item $\newsym{\protect\mathcal{N}_\alpha^+} := \{M\in\Lam{D}\ |\ \forall \beta\ge_D \alpha, \tau_\beta(M)\Da \}$, is the set of terms converging over the context $\tau_\alpha$
  \item $\newsym{\protect\mathcal{N}_\alpha^-} := \{(\sum_i\bareps_{\beta_i} + L\ |\ \alpha \ge_D \bigwedge_i\beta_i,\ L\in \Lam{D} \}$, is the set of trivial mhnf of type $\alpha$.
  \item $S_\omega:=\Lam D$ is the set of all terms.
  \item $S_\alpha:=\{G\in\Sat_D\mid \mathcal{N}_\alpha^+\supseteq G  \supseteq\mathcal{N}_\alpha^-\}$ is the set of $\tau(D)$-saturated subsets of $\mathcal{N}_\alpha^+$ that contains $\mathcal{N}_\alpha^-$ for $\alpha\in D-\{\omega\}$,
  \item $S:=(S_{\tau(D)}^\alpha)_{\alpha\in D}$ is the set of $D$-indexed collections of elements of $S_{\tau(D)}^\alpha$.
  \end{itemize}
  The definition is extended for partial models.
\end{definition}

%

\begin{lemma}\label{lemma:AdInt}
  Let $\Rea$ be a $S$-realizer in $D$.
  \begin{alignat*}4
    \text{if} &\quad (\vec \alpha,\beta)\in\llb M\rrb^{\vec x}\quad &\text{and} \quad (\forall i,L_i\in \Rea(\alpha_i))\quad &\text{then} \quad M[\vec L/\vec x]\in \Rea(\beta)\\
     \text{if} &\quad \vec \alpha\in\llb Q\rrb^{\vec x} \quad &\text{and} \quad (\forall i,L_i\in \Rea(\alpha_i))\quad &\text{then} \quad Q[\vec L/\vec x]\rta^*\epsilon
  \end{alignat*}
\end{lemma}
\begin{proof}
  By induction on $M$ and $Q$\shortv{.}\longv{:
  \begin{itemize}
  \item $M=x_i$ :
    then
    $\alpha_i\le_D \beta$. Thus if $L_i\:{\in}\: \Rea(\alpha_i)\:{\subseteq}\: \Rea(\beta)$,
    we have $M[\vec L/\vec x]=L_i\in \Rea(\beta)$.
  \item $M=N_1\ N_2$ :
    there exists $(\gamma_j,\beta_j)_{j\le n}$ such that $\beta=\bigwedge_j\beta_j$, 
    $(\vec \alpha,\bigwedge_j\gamma_j\cons\beta_j)\in\llb N_1\rrb^{\vec x}$ and $(\vec a;\bigwedge_j\gamma_j)\in\llb N_2\rrb^{\vec x}$.
    Thus, by induction hypothesis, if for all $i$, $L_i\in \Rea(\alpha_i)$, $N_1[\vec L/\vec x]\in
    (\bigcap_j(\Rea(\gamma_j)\mapsto\Rea(\beta_j)))$ and $N_2[\vec L/\vec x]\in\bigcap_j\Rea(\gamma_j)$. We conclude by $(N_1N_2)[\vec L/\vec x]\in\bigcup_j\Rea(\beta_i)$. 
  \item $M=\lambda y. N$ :
    then $\beta = \bigwedge_j \gamma_j\cons\beta_j$ and $((\vec \alpha,\bigwedge_i\gamma_i);\bigwedge_i\beta_i)\in \llb N\rrb^{\vec x y}$.
    We want to show that whenever $\forall i\le |\vec x|,\ L_i\in \Rea(\alpha_i)$ and $j\le n$, we have
    $\lambda  y.N[\vec L/\vec x]\in\ \Rea(\gamma_j)\mapsto\Rea(\beta_j)$. But if $L\in\Rea(\bigwedge_i\gamma_i)$ for all $i$, the induction hypothesis give us that for any $j$, $N[\vec L/\vec x][L/y]\in\Rea(\beta_j)$.
  \item $M=\Sigma_{j\in J}\bar\tau_{\gamma_j}(Q_j)$ :
    there is $J'\subseteq J$ such that $\beta\preceq \bigwedge_{j\in J'}\gamma_j$ and $\vec \alpha\in\bigcap_{j\in J'}\llb Q_j\rrb^{\vec x}$.\\
    By induction hypothesis, when given $\ L_i\in \Rea(\alpha_i)$ for each $i\le |\vec x|$, we get $Q_j[\vec L/\vec x]\rta^*\eps$ for any $j\in J'$.
    Thus, for all $j\in J'$, $M[\vec L/\vec x]\rta^*M'+\bareps_{\gamma_j}\in \newsym{\protect\mathcal{N}_{\gamma_j}^-}\subseteq \Rea(\gamma_j)$, so that $M[\vec L/\vec x]\in\Rea(\bigwedge_{j\in J'}\gamma_i)\subseteq\Rea(\beta)$.
  \item $Q=\tau_\beta(M)$ : we have $(\vec \alpha,\beta)\in\llb M\rrb^{\vec x}$, and by induction hypothesis if $\forall i\le |\vec x|,\ L_i\in \Rea(\alpha_i)$ then $M[\vec L/\vec x]\in \Rea(\beta)\subseteq \newsym{\protect\mathcal{N}_\beta^+}$. Thus, by definition, $\tau_\alpha(M[\vec L/\vec x])\rta^*\eps$
  \item $Q=Q_1\pt Q_2$ :
    then $\vec \alpha\in \llb Q_1\rrb^{\vec x}\cap\llb Q_2\rrb^{\vec x}$ and by induction hypothesis whenever $\forall i\le |\vec x|,\ L_i\in \Rea(\alpha_i)$,
    $Q_1[\vec L/\vec x]\rta^*\eps$ and $Q_2[\vec L/\vec x]\rta^*\eps$, thus trivially $Q_1\pt Q_2\rta^*\eps$
  \item $Q=Q_1+ Q_2$ :
    then there is $j\in\{1,2\}$, $\vec \alpha\in \llb Q_j\rrb^{\vec x}$ and by induction hypothesis whenever $\forall i\le |\vec x|,\ L_i\in \Rea(\alpha_i)$,
    $Q_j[\vec L/\vec x]\rta^*\eps$, thus trivially $Q_1\pt Q_2\rta^*\eps$
  \end{itemize}
  }
\end{proof}

\begin{theorem}\label{th:sensibility}
  A DEFiM $D$ is sensible for $\Lam{D}$ iff there is a $S$-realizer in $D$.
\end{theorem}
\begin{proof}
  Let $\Rea$ an $S^{D}$-realizer in $D$ and $\vec\alpha\in\llb Q\rrb$.
  Since for all $i\le n$, $\bareps_{\alpha_i}\in \mathcal{N}_{\alpha_i}^-\subseteq \Rea(\alpha_i)$, by Lemma~\ref{lemma:AdInt} there is $Q[\bareps_{\alpha_1}/x_1...\bareps_{\alpha_n}/x_n]\rta^*\epsilon$. In particular $Q$ is converging.\\
  Conversely, if $D$ is sensible for $\Lam{D}$, then $\Rea(\beta):=\{M\mid \exists \vec\alpha, (\vec\alpha,\beta)\in\llb M\rrb\}$ is a realizer.
\end{proof}

This means that all we have to do to prove the sensibility of a model is to look for a realizer! Unfortunately, finding such a realizer is equally difficult (which is not so surprising as both propositions are equivalent). However, if you consider that a realizer is an element of $S$ respecting the two equations of Definition~\ref{def:realizerIInD}, then we can try to make a systematic research in this set. More exactly, it is quite tempting to find such a realizer by a fixedpoint research. For this we have to turn this equations into function, but if the first one can be turned into a function using the extensionality, this is not feasible for the second one. Regardless, the second equation is natural as a structural equation and we can do our fixedpoint research inside $S^{\!\wedge}$:

\begin{lemma}\label{lemma:HnMeet}
  If we call semi $S$-realizer a function $\Rea$ such $\Rea(\alpha)\in S_\alpha$ and $\Rea(\alpha\wedge\beta)=\Rea(\alpha)\cap\Rea(\beta)$. The following function is defined over $S^{\!\wedge}$, the set of $S$-realizers:
  $$H(\Rea)(\beta)\ :=\ \bigcap_{(\gamma,\delta)\in \ext{D}(\beta)}\Bigl(\Rea(\gamma) \mapsto \Rea(\delta)\Bigr). $$
\end{lemma}
\begin{proof}
  \longv{
    if $\Rea\in S^{\!\wedge}$, then:
    \begin{itemize}
    \item For all $\alpha$, $H(\Rea)(\alpha)$ is saturated since function spaces and intersections (even infinite) of saturated set are saturated,
    \item For all $\alpha$, $\mathcal N ^-_\alpha \subseteq H(\Rea)(\alpha) \subseteq \mathcal N ^+_\alpha$ : idem,
    \item For all $\alpha,\beta$, $H(\Rea)(\alpha\wedge\beta) \subseteq H(\Rea)(\alpha)\cap H(\Rea)(\beta)$: Let $(\gamma,\delta)\in \ext{D}(\alpha)$. Since $(\gamma\cons\beta)\ge_D\bigwedge_{(\gamma',\beta')\in \ext{D}(\alpha\wedge\beta)}(\gamma'\cons\beta')$, we can use the distributivity to get a decomposition $\delta=\bigwedge_i \delta_j$ such that for all $i$, $\gamma\cons\delta_i\ge_D \gamma_i'\cons\delta_i'$ for some $(\gamma_i',\delta_i')\in \ext{D}(\alpha\wedge\beta)$. This means that $\Bigl(\Rea(\gamma) \mapsto \Rea(\delta)\Bigr)=\bigcup_i\Bigl(\Rea(\gamma) \mapsto \Rea(\delta_i)\Bigr)\subseteq \bigcup_i\Bigl(\Rea(\gamma_i') \mapsto \Rea(\delta_i')\Bigr)$ since $\Rea(\gamma)\supseteq\Rea(\gamma'_i)$ and $\Rea(\delta_i)\subseteq\Rea(\delta_i')$, we conclude since each  $(\gamma_i',\delta_i')\in \ext{D}(\alpha\wedge\beta)$. 
    \item For all $\alpha,\beta$, $H(\Rea)(\alpha\wedge\beta) \supseteq H(\Rea)(\alpha)\cap H(\Rea)(\beta)$: Let $(\gamma,\delta)\in \ext{D}(\alpha\wedge\beta)$. Since $(\gamma\cons\beta)\ge_D\bigwedge_{(\gamma',\beta')\in \ext{D}(\alpha)\cup \ext{D}(\beta)}(\gamma'\cons\beta')$, we can use the distributivity to get a decomposition $\delta=\bigwedge_i \delta_j$ such that for all $i$, $\gamma\cons\delta_i\ge_D \gamma_i'\cons\delta_i'$ for some $(\gamma_i',\delta_i')\in \ext{D}(\alpha)\cup \ext{D}(\beta)$. This means that $\Bigl(\Rea(\gamma) \mapsto \Rea(\delta)\Bigr)=\bigcup_i\Bigl(\Rea(\gamma) \mapsto \Rea(\delta_i)\Bigr)\subseteq \bigcup_i\Bigl(\Rea(\gamma_i') \mapsto \Rea(\delta_i')\Bigr)$ since $\Rea(\gamma)\supseteq\Rea(\gamma'_i)$ and $\Rea(\delta_i)\subseteq\Rea(\delta_i')$, we conclude since each  $(\gamma_i',\delta_i')\in \ext{D}(\alpha)\cup \ext{D}(\beta)$.
    \end{itemize}
  }
  \shortv{By using the distributivity condition on our DEFiM.}
\end{proof}

Now, all we need is to find a fixedpoint... which easier said than done. In fact, interesting examples will have to be dealt using strong fixedpoint theorems. Indeed, fixedpoint {\em \`a la} Curry are not sufficient, even Tarski's fixedpoint are often not enough. Among order theoretic fixedpoint theorems, the following version is the most general that the author could find.\footnote{To the author knowledge, it is the first time it has been enunciated formally.}

\begin{definition}
  The {\em lexicographic stratification} of a set $X$ is a sequence $(\sqsubseteq_n)_{n\in\kappa}$ of preorders, for $\kappa$ is any cardinal, verifying:
  \begin{itemize}
  \item $\bigcap (\equiv_n)$ is the equality in $X$, where $(\equiv_n):=(\sqsubseteq_n)\cap(\sqsupseteq_n)$,
  \item for any $n$ in $\kappa$, $(\sqsubseteq_n)\subseteq (\equiv_{\downarrow n})$, where
    $ (\equiv_{\downarrow n})\ :=\ \displaystyle\bigcap_{m<n} (\equiv_m)\ ,$
  \item for all $U\in \bigslant X {\equiv_{\downarrow n}}$, the poset $(\bigslant U {\equiv_n} , \sqsubseteq_n)$ is a dcpo.
  \end{itemize}
  
  \noindent
  A function $f$ on such a stratification is {\em lexicographically-monotonous} whenever:
  \begin{itemize}
  \item $f$ respect the equivalences $(\equiv_{\downarrow n})$, {\em i.e.}, for any $n\in \kappa$ and any pair $x\longv{,}\shortv{\equiv_{\downarrow n}}y$\longv{:
    $$ (x \equiv_{\downarrow n} y) \quad \Rta \quad (f(x) \equiv_{\downarrow n} f(y))\ ,$$
    }\shortv{ we have $f(x) \equiv_{\downarrow n} f(y)$,}
  \item $f$ is ${\downarrow }n$-monotonous over $(\equiv_{\downarrow m})_{m \prec n}$-fixedpoint, {\em i.e.}, for any $n\longv{\in \kappa}$ and any pair $x\equiv_{\downarrow}y \in X$\longv{:
    $$ f(x) \equiv_{\downarrow n} x \quad \Rta\quad \left(\ x\sqsubseteq_n y \quad \Rta \quad f(x) \sqsubseteq_n f(y)\ \right)\ . $$
    }\shortv{ if $f(x) \equiv_{\downarrow n} x$ and $x\sqsubseteq_n y$ then $f(x) \sqsubseteq_n f(y)$.}
  \end{itemize}
\end{definition}


\begin{prop}
  Any lexicographically-monotonous function on a lexicographically-stratified set has a fixedpoint.
\end{prop}
\begin{proof}\longv{\todo{rewrite with more details}}
  By induction on $n\in\kappa$. Suppose given $X_{\downarrow n}\in \bigslant S {\equiv_{\downarrow n}}$ such that $f(X_{\downarrow n}) \equiv_{\downarrow n} X_{\downarrow n}$, then $f$ make sens and is monotonous in the dcpo $(\bigslant {X_{\downarrow n}} {\equiv_n} , \sqsubseteq_n)$. Thus it has a least fixedpoint $X_n$. Notice that $X_{\downarrow n}\supseteq X_n$ so that we can take limits. In the end, we get a fixedpoint $X_{\downarrow\kappa}\in \bigslant D {\bigcup\!\!\equiv_n} = D$. 
\end{proof}

Now that we have our fixedpoint theorem, we have to link it to the considered filter model and stratify $S^{\!\wedge}$. Since we are looking for a condition on the atoms (or the intersection types) of our model, it is only natural to try to stratify $S^{\!\wedge}$ along those. However, this may be a bit arbitrary, which in turn may be one of the reason of our ultimate incompleteness...

\begin{definition}
A preorder $(D,\preceq)$ is said \newdefprem{well founded}{preorder} if the quotiented poset $\bigslant{(D,\preceq)}\simeq$ over the induced equivalence $\simeq := (\preceq\cap\succeq)$ is well founded. It is said total if any two element are comparable.
\end{definition}

\begin{definition}
  A DEFiM $D$ is said \emph{$S$-realizable by stratification} if 
  \begin{itemize}
  \item for every $\alpha\in D$, there is a dcpo $(\subseteq_\alpha)$ over $S_{\alpha}$,
  \item there is a total and well founded preorder $(S,\preceq)$ on $D$, 
  \item $S^{\!\wedge}$ is lexicographically stratified by $(\sqsubseteq_{a})_{a\in D\raisebox{0em}{$\!_{/\!\simeq}$}}$ defined by:
    $$ \Rea\sqsubseteq_{[\alpha]} \Reb \quad\text{iff} \quad 
    \left\{
      \begin{matrix} 
        \forall\beta\prec\alpha,\ &\Rea(\beta) = \Reb(\beta) \\  
        \forall\beta\simeq\alpha,\ &\Rea(\beta) \subseteq_\beta \Reb(\beta)
      \end{matrix}
    \right.
    $$
  \item $H$ is lexicographically-monotonous.
  \end{itemize}
\end{definition}

\begin{remark}
  \begin{itemize}
  \item Remark that $H$ may not be monotonous, and will not be in general.
  \item More important, notice that for $(\sqsubseteq_{a})_{a\in D\raisebox{0em}{$\!_{/\!\simeq}$}}$ to be a stratification, we only need to prove the last condition; {\em i.e.}, that for all $X\in \bigslant {S^{\!\wedge}} {\equiv_{\downarrow a}}$, the poset $(\bigslant X {\equiv_a} , \sqsubseteq_a)$ is a dcpo. This property says that for any sequence $(\Rea(\beta))_{\beta\prec\alpha}\in (S_{\tau(D)}^{\beta})_{\beta\prec\alpha}$ that can be extended as an element of $S^{\!\wedge}$, the set of possible extensions for the class $a$ forms a dcpo.
  \item Assuming the axiom of choice, the preorder $\preceq$ may not have to be total.
  \end{itemize}
\end{remark}

\begin{theorem}
  Any DEFiM $D$ that is $S$-realizable by stratification has a $S$-realizer in $D$.
\end{theorem}

\newcommand\valuation[1]{\mathcal V (#1)}

\subsection{Positive stratification}

The notion of ``realizability by stratification'' is still too abstract; it particular, it intrinsically refers to syntactical aspects of the considered calculi. We had like a property only referring to the internal structure of the type system without any syntactic notion. 

In order to achieve this goal, we need yet another change of perspective, which in turn introduce yet another source of arbitrary. Nonetheless, positive stratification include all filter models proven sensible in the literature. We will discuss at the end of those that are conjectured sensible but not proven by lake of adequate techniques.

\begin{definition} \label{def:wpos}
  A (partial) DEFiM $D$ is \newdef{stratified positive} (SP for short) if there exist 
  \begin{itemize}
  \item a valuation $\mathcal{V}$, called polarity, from $D-\{\omega\}$ in the Booleans $\{\true,\false\}$, 
  \item a well founded and total preorder $\preceq$ in $D$ with $\omega$ as a bottom, 
  \end{itemize}
  such that for all $\gamma\in D$ and all $(\alpha,\beta)\in \ext{D}(\gamma)$:
  \begin{align*}
    \gamma &\succeq \beta,   & 
    \gamma \simeq \beta\ &\Rta\ \mathcal{V}(\gamma)=\mathcal{V}(\beta), \\
    \gamma &\succeq \alpha,    &
    \gamma \simeq \alpha\ &\Rta\ \mathcal{V}(\gamma)\neq \mathcal{V}(\alpha),
  \end{align*}
  (where $\simeq := (\preceq\cap\succeq)$ is the equivalence relation induced by the preorder)\\
  and such that:
  \begin{align*}
    \alpha\wedge\beta&\preceq \gamma \text{ for $\gamma=\alpha$ or for $\gamma=\beta$} &  (\alpha\wedge\beta) \prec \alpha  &\Rta (\alpha\wedge\beta) = \beta
  \end{align*}
  Moreover, we also require that the polarity is coherent with the intersections on $\simeq$-equivalence classes:
  $$ \alpha\simeq\beta \quad \Rta\quad \valuation{\alpha\wedge\beta} = \valuation \alpha\wedge\valuation\beta.$$
\end{definition}

This condition can be seen as a stratification given by $\preceq$, where the quotient $\bigslant D {\simeq}$ represents the different levels of the stratification, each level endowed with a positive polarity $\mathcal{V}$. This stratification improves the condition of~\cite{Ber00} that only considers completions of positive partial DEFiM.\footnote{More exactly it considers a subclass of DEFiM called K-models.} This condition is the invariant by completion, which simplify the proof of stratified positivity of DEFiMs of Example~\ref{example:1} (save for $P_\infty$).

\begin{prop}
  Assuming the axiom of choice, in the definition of stratified positive DEFiM, the preorder $\preceq$ can be taken non-total without lost of generality.
\end{prop}

\begin{prop}\label{prop:weakposmod}
  A partial DEFiM $E$ is stratified positive iff its completion $\Comp{E}$ is stratified positive.
\end{prop}

\begin{example}\label{ex:WP}
  The models of Example~\ref{example:1} are stratified positive except $P_\infty$ and $U_\infty$:
  \begin{itemize}
  \item \Dinf is SP: The stratified positivity is given by $\mathcal{V}(*)=\false$ and $\omega\prec *$.
  \item $\Dinf^*$ is SP: Idem, we set $\mathcal{V}(q)=\true$, $\mathcal{V}(p)=\false$ and $\omega\prec p\simeq q$.
  \item $Z_\infty$ is SP: Idem, we set $\mathcal{V}(\underline{2n})=\true$, $\mathcal{V}(\underline{2n+1})=\false$ and $\omega\prec \underline m\simeq \underline n$ for all $m$ and $n$.
  \item $P_\infty$ is not SP: Since $*\cons * = *$, they are $\preceq$-equivalent and with the same polarity, contradicting the second implication in Definition~\ref{def:wpos}.
  \item $U_\infty$ is not SP: Since $\underline n = \underline {n\+1} \cons \underline {n\+1}$, we must have $\underline {n} \succ \underline {n\+1}$, which creates a non well-founded chain.
  \end{itemize}
\end{example}



\begin{lemma}\label{lemma:dcpo}
  Let $\Rea\in (S_\alpha)_{\alpha\prec \delta}$ such that $\Rea(\alpha\wedge\beta)=\Rea(\alpha)\cap\Rea(\beta)$ for $\alpha,\beta\prec \delta$.\\
  The set of extensions of $\Rea$ to all $\alpha\preceq\delta$, ordered by $\sqsubseteq_{\mathcal V}$, is a dcpo with a sup $(\bigvee_i\Rea_i)(\alpha)$ defined by induction on~$\preceq$: 
  $$
  (\bigvee_i\Rea_i)(\alpha) = 
  \left\{
    \begin{matrix} 
      \mathcal N^-_\alpha \cup \displaystyle\bigcup_{\gamma\le_D \alpha, \gamma\prec\delta}\Rea(\gamma) \cup \bigcup\Rea_i(\alpha) &\text{ whenever }\mathcal V(\alpha)=\false,\\[0.2em]
      \mathcal N^+_\alpha \cap \displaystyle\bigcap_{\gamma\ge_D \alpha, \gamma\prec\delta}\Rea(\gamma) \cap \bigcap\Rea_i(\alpha) &\text{ whenever } \mathcal V(\alpha)=\true.
    \end{matrix}
  \right.
  $$
  in particular, $(\bigvee_i\Rea_i)(\alpha)=\Rea(\alpha)$ for $\alpha\prec \delta$.
\end{lemma}
\longv{
\begin{proof}
   We first show that for all $\alpha\le_D\beta$, then $(\bigvee_i\Rea_i)(\alpha)\subseteq (\bigvee_i\Rea_i)(\beta)$.
   \begin{itemize}
   \item if $\mathcal V(\beta)=\false$:
      The case where $\alpha\prec \delta$ is trivial (it is the second terms of the definition above). Otherwise, necessarily $\mathcal V(\alpha)=\false$:
      We have $\mathcal V(\beta)=\mathcal V(\alpha)$ thus we only have to check term to term. First, we have $\mathcal N^-_\alpha\subseteq \mathcal N^-_\beta$. For the second term, we have that $\{\gamma\mid \gamma\le_D\alpha,\gamma\prec\delta\}\subseteq\{\gamma\mid \gamma\le_D\beta,\gamma\prec\delta\}$, thus $\displaystyle\bigcup_{\gamma\le_D \alpha, \gamma\prec\delta}\Rea(\gamma)\subseteq \displaystyle\bigcup_{\gamma\le_D \beta, \gamma\prec\delta}\Rea(\gamma)$. And last, we have $\Rea_i(\alpha)\subseteq\Rea_i(\beta)$ for all $i$. 
   \item if $\mathcal V(\alpha)=\true$:
     The case where $\beta\prec \delta$ is trivial (it is the second terms of the definition above). Otherwise, necessarily $\mathcal V(\beta)=\true$:
     We have $\mathcal V(\beta)=\mathcal V(\alpha)$ thus we only have to check term to term. First, we have $\mathcal N^+_\alpha\subseteq \mathcal N^+_\beta$. For the second term, we have that $\{\gamma\mid \gamma\ge_D\alpha,\gamma\prec\delta\}\supseteq\{\gamma\mid \gamma\ge_D\beta,\gamma\prec\delta\}$, thus $\displaystyle\bigcap_{\gamma\ge_D \alpha, \gamma\prec\delta}\Rea(\gamma)\subseteq \displaystyle\bigcap_{\gamma\ge_D \beta, \gamma\prec\delta}\Rea(\gamma)$. And last, we have $\Rea_i(\alpha)\subseteq\Rea_i(\beta)$ for all $i$. 
   \item if $\mathcal V(\alpha)=\false$ and $\mathcal V(\beta)=\true$: We have $\mathcal N^-(\alpha)\subseteq \mathcal N^-_\beta\subseteq (\bigvee\Rea)(\beta)$. Similarly, $(\bigvee\Rea)(\alpha)\subseteq\mathcal N^+_{\alpha}\subseteq\mathcal N^+_\beta$. For any $\gamma_+,\gamma_-\prec \delta$ such that $\gamma_+\le_D\alpha\le_D\beta\le_D\gamma_-$, we have $\Rea(\gamma_+)\subseteq\Rea(\gamma_-)$. For any $\gamma\prec \delta$ such that $\gamma\le_D\alpha\le_D\beta$ and any $i\in I$, $\Rea(\gamma)=\Rea_i(\gamma)\subseteq \Rea_i(\beta)$. Similarly, for any $\gamma\prec \delta$ such that $\gamma\ge_D\beta\ge_D\alpha$ and any $i\in I$, $\Rea(\gamma)=\Rea_i(\gamma)\supseteq \Rea_i(\alpha)$. The only remaining case is for each $i,j\in I$, to prove that $\Rea_i(\alpha)\subseteq\Rea_j(\beta)$, but we know that $\Rea_i(\alpha)\subseteq\Rea_{i\vee j}(\alpha)$ since $\valuation\alpha=\false$, similarly, $\Rea_{i\vee j}(\beta)\subseteq\Rea_j(\beta)$ since $\valuation\beta=\true$, and we conclude by $\Rea_{i\vee j}(\alpha)\subseteq\Rea_{i\vee j}(\beta)$ since $\alpha\le_D\beta$.
   \end{itemize}
   Now, we have to verify that all meets are conserved. One inclusion is already done, so that we have to show that $(\bigvee_i\Rea_i)(\alpha)\cap(\bigvee_i\Rea_i)(\beta)\subseteq (\bigvee_i\Rea_i)(\alpha\wedge\gamma)$. Moreover, the cases where  $\alpha,\beta\prec\delta$, $\alpha =(\alpha\wedge \beta)$ or $\beta= (\alpha\wedge\beta)$ are trivial, thus we assume that $\alpha\preceq \beta\simeq(\alpha\wedge \beta)\simeq\delta$:
   \begin{itemize}
   \item If $\valuation {\alpha\wedge\beta}=\true$:\\
     Then necessarily $\valuation \beta = \true$.  We have $\mathcal N^+_{\alpha\wedge\beta}\supseteq\mathcal N^+_\alpha\cap\mathcal N^+_\beta\supseteq(\bigvee\Rea)(\alpha)\cap(\bigvee\Rea)(\beta)$. Moreover, for any $\gamma\prec\delta$ such that $\gamma\ge_D\alpha\wedge\beta$, we have $\Rea(\gamma)=(\bigvee\Rea)(\gamma)\supseteq(\bigvee\Rea)(\alpha\wedge\beta)$. Finally, we got the difficult case: let $i\in I$, we have $\Rea_i(\alpha\wedge\beta)=\Rea_i(\alpha)\cap\Rea_i(\beta)$ since $\Rea_i$ respect intersections, and we have $(\bigvee\Rea)(\beta)\subseteq\Rea_i(\beta)$ since $\valuation \beta = \true$, we thus need to show that $(\bigvee\Rea)(\alpha)\subseteq\Rea_i(\alpha)$ to get $\Rea_i(\alpha\wedge\beta)\supseteq(\bigvee\Rea)(\alpha)\cap(\bigvee\Rea)(\beta)$; there is two cases:
     \begin{itemize}
     \item either $\alpha\simeq\delta$: then necessarily $\valuation \alpha = \true$, so that  $(\bigvee\Rea)(\alpha)\subseteq\Rea_i(\alpha)$,
     \item or $\alpha\prec\delta$: then $(\bigvee\Rea)(\alpha)=\Rea(\alpha)=\Rea_i(\alpha)$.
     \end{itemize}
   \item If $\valuation {\alpha\wedge\beta}=\false$:\\
     We can the consider that $\valuation\beta = \false$ without lost of generality.\footnote{It is also possible that we only have $\alpha\simeq\beta$ and $\valuation \alpha = \false$, but we can the conclude by symmetry.}
     Notice that $\Rea(\alpha)\cap\mathcal N^-_\beta \subseteq \mathcal N^-_\alpha\cap\mathcal N^+_\beta$ which is included in the completion of $\mathcal N^-_{\alpha\wedge\beta}$ and thus in $(\bigvee\Rea)(\alpha\wedge\beta)$.
     \begin{itemize}
     \item If $\alpha\prec \delta$: Then for all $\gamma\prec\delta$ such that $\gamma\le\beta$, $(\bigvee\Rea)(\alpha)\cap\Rea(\gamma)=\Rea(\alpha)\cap\Rea(\gamma)=\Rea(\alpha\wedge\gamma)\subseteq (\bigvee\Rea)(\alpha\wedge\beta)$ the last inclusion being because $\alpha\wedge\gamma\le_D\alpha\wedge\beta$. Moreover, for any $i\in I$, $(\bigvee\Rea)(\alpha)\cap\Rea_i(\beta)=\Rea_i(\alpha)\cap\Rea_i(\beta)=\Rea_i(\alpha\wedge\beta)\subseteq(\bigvee\Rea)(\alpha\wedge\beta)$ the last inclusion being because $\valuation{\alpha\wedge\beta}=\false$.
     \item If $\alpha\simeq \delta$ and $\valuation\alpha = \false$: For all $\gamma_1,\gamma_2\prec\delta$ such that $\gamma_1\le_D\alpha$ and $\gamma_2\le_D\beta$, we have $\Rea(\gamma_1)\cap\Rea(\gamma_2)=\Rea(\gamma_1\wedge\gamma_2)\subseteq(\bigvee\Rea)(\alpha\wedge\beta)$, the last inclusion being because $\alpha\wedge\gamma\le_D\alpha\wedge\beta$. Moreover, for all $\gamma\prec\delta$ such that $\gamma\le_D\alpha$ and all $i\in I$, $\Rea(\gamma)\cap\Rea_i(\beta)=\Rea_i(\gamma)\cap\Rea_i(\beta)=\Rea_i(\gamma\wedge\beta)\subseteq \Rea_i(\alpha\wedge\beta)\subseteq(\bigvee\Rea)(\alpha\wedge\beta)$. Finally, for any $i,j\in I$, we have $\Rea_i(\alpha)\cap\Rea_j(\beta)\subseteq\Rea_{i\vee j}(\alpha)\cap\Rea_{i\vee j}(\beta)=\Rea_{i\vee j}(\alpha\wedge\beta)\subseteq (\bigvee\Rea)(\alpha\wedge\beta)$.
     \item If $\alpha\simeq \delta$ and $\valuation\alpha = \true$: Then for all $i\in I$, $(\bigvee\Rea)(\alpha)\cap\Rea_i(\beta)\subseteq \Rea_i(\alpha)\cap\Rea_i(\beta)=\Rea_i(\alpha\wedge\beta)\subseteq(\bigvee\Rea)(\alpha\wedge\beta)$, the first inclusion being because $\valuation\alpha = \true$. Moreover, for any $\gamma\prec\delta$ such that $\gamma\le_D\beta$, we have seen that $(\bigvee\Rea)(\alpha)\cap(\bigvee\Rea)(\gamma)\subseteq (\bigvee\Rea)(\alpha\wedge\gamma)\subseteq(\bigvee\Rea)(\alpha\wedge\beta)$.
     \end{itemize}
   \end{itemize}
\end{proof}
}

\begin{lemma}
  Any stratified positive DEFiM $D$ is $S$-realizable by stratification.
\end{lemma}
\begin{proof}
  \longv{
  \begin{itemize}
  \item For any $\alpha \in D$ we define the order $(\subseteq_\alpha):=(\subseteq_{\mathcal V(\alpha)})$ where $(\subseteq_\false):=(\subseteq)$ and $(\subseteq_\true):=(\supseteq$), so that $(S_{\alpha}, \subseteq_\alpha)$ is a dcpo.
  \item The equivalence classes $\bigslant D \simeq$ forms a $J$-partition of $D$ for $J$ the cardinal of $\bigslant D \simeq$.
  \item $S^{\!\wedge}$ is lexicographically stratified by $(\sqsubseteq_{a})_{a\in D\raisebox{0em}{$\!_{/\!\simeq}$}}$ defined by:
    $$ \Rea\sqsubseteq_{[\alpha]} \Reb \quad\text{iff} \quad 
    \left\{
      \begin{matrix} 
        \forall\beta\prec\alpha,\ &\Rea(\beta) = \Reb(\beta) \\  
        \forall\beta\simeq\alpha,\ &\Rea(\beta) \subseteq_\beta \Reb(\beta)
      \end{matrix}
    \right.
    $$
    We only need to prove that for all $U\in \bigslant X {\equiv_{\downarrow n}}$, the poset $(\bigslant U {\equiv_n} , \sqsubseteq_n)$ is a dcpo; which corresponds to Lemma~\ref{lemma:dcpo}
  \item Remains to show that $H$ is lexicographically-monotonous:
    \begin{itemize}
    \item $H$ respects the equivalences $(\equiv_{\downarrow c})$:\\
      Let $\alpha\in D$ and $\Rea \equiv_{\downarrow [\alpha]} \Reb$. Let $\beta\ge_D \alpha$, we have $H(\Rea)(\beta)\ :=\ \bigcap_{(\gamma,\delta)\in \ext{D}(\beta)}\Bigl(\Rea(\gamma) \mapsto \Rea(\delta)\Bigr)$ and $H(\Reb)(\beta)\ :=\ \bigcap_{(\gamma,\delta)\in \ext{D}(\beta)}\Bigl(\Reb(\gamma) \mapsto \Reb(\delta)\Bigr)$. It is sufficient to show that $\Bigl(\Rea(\gamma) \mapsto \Rea(\delta)\Bigr) = \Bigl(\Reb(\gamma) \mapsto \Reb(\delta)\Bigr)$ for any $(\gamma,\delta)\in \ext{D}(\beta)$. But this is immediate since $\gamma,\delta\preceq \beta\preceq\alpha$ and $\Rea \equiv_{\downarrow [\alpha]} \Reb$. 
  \item $H$ is ${\downarrow }[\alpha]$-monotonous over $(\equiv_{\downarrow [\beta]})_{\alpha \prec \beta}$-fixedpoint:\\
    Let $\alpha\in D$ and $\Rea \sqsubseteq_{\downarrow [\alpha]} \Reb$ such that $H(\Rea)(\beta) = \Rea(\beta)$ for all $\beta\prec\alpha$.
    For any $\beta\prec\alpha$, we have $H(\Rea)(\beta) = H(\Reb)(\beta)$ since $H$ respects the equivalences $(\equiv_{\downarrow c})$.
    Remains to show that for all $\beta\simeq\alpha$, $H(\Rea)(\beta) \subseteq_\beta H(\Reb)(\beta)$. We will show that for any $(\gamma,\delta)\in \ext{D}(\beta)$, $\Bigl(\Rea(\gamma) \mapsto \Rea(\delta)\Bigr) \subseteq_\beta \Bigl(\Reb(\gamma) \mapsto \Reb(\delta)\Bigr)$. We do the case where $\valuation\beta=\false$, the other is symmetric. Either $\gamma\prec \beta$ and $\Rea(\gamma)=\Reb(\gamma)$ (since $\Rea\equiv_{\downarrow[\gamma]}\Reb$) or $\gamma\simeq\beta$ has the polarity $\valuation\gamma=\true$ and $\Rea(\gamma)\supseteq\Reb(\gamma)$, in any case, $\Rea(\gamma)\supseteq\Reb(\gamma)$. Similarly, in any case $\Rea(\delta)\subseteq\Reb(\delta)$, so that we have $\Bigl(\Rea(\gamma) \mapsto \Rea(\delta)\Bigr) \subseteq_\beta \Bigl(\Reb(\gamma) \mapsto \Reb(\delta)\Bigr)$.
    \end{itemize}
  \end{itemize}
  }
  \shortv{
    We define  $(\subseteq_\alpha):=(\subseteq_{\mathcal V(\alpha)})$ where $(\subseteq_\false):=(\subseteq)$ and $(\subseteq_\true):=(\supseteq$).
    After Lemma~\ref{lemma:dcpo}, we can lexicographically stratify $S^{\!\wedge}$ by $ \Rea\sqsubseteq_{[\alpha]} \Reb$ if 
    for all $\beta\prec\alpha$, $\Rea(\beta) = \Reb(\beta)$ and for all $\beta\simeq\alpha$, $\Rea(\beta) \subseteq_\beta \Reb(\beta)$.
    The lexicographical-monotonicity of $H$ follows.
  }
\end{proof}

\begin{theorem}\label{th:sensibility}
  Any stratified positive DEFiM $D$ is sensible for $\Lam{D}$ and approximable.
\end{theorem}

\begin{example}\label{ex:sensWT}
  By Theorem~\ref{th:sensibility} and Example~\ref{ex:WP}, all the DEFiMs $D$ of Examples~\ref{example:1} \longv{and~\ref{ex:coind}} are approximable except for $P_\infty$ and $U_\infty$.
\end{example}

\longv{

\subsection{Further generalization}

We strongly conjecture that this result does not fundamentally use the extensionality:

\begin{conjecture}\label{conj}
  Any stratified positive filter model $D$ is approximable.
\end{conjecture}

This result should be obtained following the same way, but with a lot of technical hindrance. In particular the rules $(\tau)$ and $(\bar\tau)$ would become potentially infinitary:\footnote{In the sens that sum and product could be infinite.}
 \begin{align*}
   (\tau)     &&    \tau_\alpha(\lambda x.M) &\rta \sum_{A\in \mathcal A}\prod_{(\beta,\gamma)\in A}\tau_\gamma(M[\bareps_\beta/x]) &\text{where } \mathcal A =\left\{A\subseteq_f D\times D\ \middle|\ \bigwedge_{(\beta,\gamma)\in A}(\beta\cons\gamma)\le\alpha\right\} \\
   (\bar\tau) &&    (\sum_i\bar\tau_{\alpha_i}(Q_i)) N &\rta \sum_i\sum_{(\beta\cons\gamma)\ge\alpha_i} \bar\tau_{\gamma_i}(Q\pt\tau_{\beta_i}(N)) 
 \end{align*}
Another technical issue is the definition of the function $H$ of Lemma~\ref{lemma:HnMeet} that would be no more a function, but just linear constraints.

This generalization is expected for the long version; especially because it surprisingly permit to weaken the condition positive stratification by dropping the well foundedness of the strata.

\begin{proposition}
  Let $D$ a filter model satisfying all the conditions of stratified positiveness except for the well foundedness of the preorder $\preceq$.\\
  If Conjecture~\ref{conj} is true, then $D$ equates any terms with the same B\"ohm trees, and is in particular sensible.
\end{proposition}
\begin{proof}
  Let $M$ and $N$ two terms with the same set of B\"ohm approximations and let $(\vec\alpha,\beta)\in\llb M\rrb_D^{\vec x}$. We will show that  $(\vec\alpha,\beta)\in\llb N\rrb_D^{\vec x}$.

  There exists a derivation $\pi$ of $(x_i:\alpha_i)_i\vec M:\beta_i$ in the intersection type system of $D$. Since $\pi$ is finite, there is only a finite set $F\subseteq_f D$ of elements of $D$ appearing in the derivation.

  Let $F^\wedge\subseteq D$ the $\wedge$-completion $F^\wedge:=\{\bigwedge_i \gamma_i\mid \forall i, \gamma_i\in F\}$ of $F$. Let $\rta_F$ partially defined by $\gamma\rta_F\delta = \gamma\rta\delta$ when it makes sens, {\em i.e.}, when $\gamma,\delta,(\gamma\rta\delta) \in F^\wedge$. 

  Then $(F^\wedge, \wedge,\rta_F)$ is a partial filter model that can be freely completed into $\overline F$. Moreover, $(F^\wedge, \wedge,\rta_F)$ is stratified positive since it is finite and a subset of $D$; thus $\overline F$ is stratified positive.

  Since $\pi$ only use elements of $F$, it is also a derivation in $\overline F$, so that $(\vec\alpha,\beta)\in\llb M\rrb_{\overline F}^{\vec x}$. Since $\overline F$ is stratified positive and $M$ and $N$ have the same set of B\"ohm approximations, $(\vec\alpha,\beta)\in\llb N\rrb_{\overline F}^{\vec x}$. Moreover, since $D$ and $\overline F$ are two completions of $(F^\wedge, \wedge,\rta_F)$ but $\overline F$ is free, we have $\llb.\rrb_{\overline F}\subseteq \llb .\rrb_D$; so that $(\vec\alpha,\beta)\in\llb N\rrb_{D}^{\vec x}$.
\end{proof}

\begin{remark}
  Equating all terms with the same B\"ohm trees is a notion similar to approximability, but slightly weaker. This is a property that says that the interpretation of a term is characterized by the interpretations of its B\"ohm trees; but it may not be the union that is considered. Morally, however, this is a kind of approximation theorem where the ``limit'' of the interpretations can be arbitrary (and not just the union).
\end{remark}

\begin{example}\label{ex:sensWT}
  Assuming Conjecture~\ref{conj}, the filter model $U_\infty$ of Examples~\ref{example:1} equates any terms with the same B\"ohm trees.
\end{example}

}

\longv{
%

\subsection*{Related Works}

The quest for sensibility and approximability of different filter models was very important in the 90's. A survey of this quest can be found in the book ``Lambda calculus with types''~\cite[Chapter 17]{BarDekSta}. 

We only have one reference to add to their survey, this is the works of Berline~\cite{Ber00} and her students Guy~\cite{Gou95}, Kerth~\cite{Ker98} and Manzonetto~\cite{Man09}. They performed deep studies on the limits and classification of the traditional classes of models. In that aspect, they follow an approach very similar to ours.

As a systematic study of a specific property in a large class of models, this article also follows recent works of Breuvart, Manzonetto and Ruopolo~\cite{Bre14,Bre16Sem,BrMaPoRu2016} that are rather studying the property of full abstraction for different reduction strategies. 

Indirectly, the (relatively) recent results of Ehrhard on the extensional collapse~\cite{Ehr09} are also linked with our result as the target of the described extensional collapse are automatically approximable (because the source is a class containing only approximable models). This gives yet a different and modern approach of approximability.

\subsection*{Further Works}

One may ponder the generality of our work considering the restriction taken on our class of model. First, the choice of filter models over usual Scott domains seems relatively safe as a Scott domain can be turned into a filter model by adding a top element; in the other side not having to consider the existence of an intersection is before all a comfort for the reader. Moreover, switching to Scott domains would make heavier the definition of tests, similarly for the others enforced restrictions: the extensionality and the distributivity. We strongly believe that the detour by tests mechanism can be removed, removing these unnatural restrictions. Nonetheless, we choose to stick with tests as they illustrate the link between sensibility and approximability in a very readable manner.

Our main regret, however, is that the final characterization is not a complete one: there is ({\em a priori}) filter models that are approximable and not positively stratifyable, or even models that are sensible but not appriximable! To illustrate this remark, we look at four filter models that are generated by the atoms $\alpha,\beta,\gamma,\delta$ and the following four sets of equations:\footnote{In the first tree systems, $\alpha=\beta$ there is just only three atoms.}
  \begin{align}
    \alpha &= \omega\cons\alpha &
    \beta &= \omega\cons\alpha &
    \gamma &= (\gamma\wedge\delta)\cons\beta &
    \delta &= \omega\cons\omega\cons\alpha\\
    \alpha &= \omega\cons\alpha &
    \beta &= \omega\cons\alpha &
    \gamma &= (\gamma\wedge\delta)\cons\beta &
    \delta &= \alpha\cons\alpha\cons\alpha\\
    \alpha &= \omega\cons\alpha &
    \beta &= \omega\cons\alpha &
    \gamma &= (\gamma\wedge\delta)\cons\beta &
    \delta &= \omega\cons\alpha\cons\alpha\\
    \alpha &= \omega\cons\alpha &
    \beta  &= (\beta\cons\alpha)\cons\alpha \!\!\!\! &
    \gamma &= (\gamma\wedge\delta)\cons\beta  &
    \delta &= \omega \cons \alpha\cons\alpha
  \end{align}
Notice that the notation $\omega\cons\omega\cons\alpha$ is simply syntactic sugar for $\omega\cons\gamma'$ for $\gamma'=\omega\cons\alpha$. Considering that we omit the full description of $\wedge$ and $\ext{D}$ which are the free ones, each of these lines forms a partial DEFiM.

In the first model, $\delta\le\gamma$ since $\delta=\omega\cons\alpha$ and $\gamma=(\gamma\wedge\delta)\cons\alpha$ with $\omega\ge (\gamma\wedge\delta)$ (remember that $\omega$ is a top). Thus the equation $\gamma=\delta\cons\alpha$ is now positive, and the generated model is positively stratified. 

On the other hand, in the second model, $\delta\ge\gamma$; thus $\gamma=\gamma\cons\beta$ is an unsafe equation breaking sensibility because $\gamma\in\llb\Omega\rrb$. The third one is more interesting; in this case, neither $\delta\ge\gamma$ not $\delta\le\gamma$; it is conjectured that this model is sensible and approximable but no proof have be found yet. 

The last example is even more surprising: it is also conjectured sensible for the same reason, but it can be shown non-approximable. This is an example that appears\footnote{In a slightly more complex form} in Kerth's thesis~\cite{KerthPhD}, he showed (more or less) that if we consider the $\lambda$-term $ V\: :=\: (\lambda xy.y(xx))\ (\lambda xy. y(xx)) $, then $\beta$ is in the interpretation of $V$, but $\tau_\beta(V)$ diverges. None of these two facts are difficult to obtain and we invite our reader to verify it as an exercise.

\subsection*{Conclusion}

With this highly theoretical and exploratory article, we only aim at questioning the limits of our models by pointing on unusual behaviors of well known semantical objects. 

Indeed, we have seen that approximability and sensibility are properties that are surprisingly hard to separate by traditional filter models. The possible causes are easy to see:
\begin{itemize}
\item Either it may rise from a new internal incompleteness of the considered class of model, which would join the incompleteness of~\cite{CaSa09}.
\item But it is more probably a logical weakness of the methods we know for proving the sensibility of a model.
\end{itemize}

In the second case, this would be an indication that the realizability methods are in fact limited when joining coinductive types and subtyping. It is, however, impossible to discern at which point level is the blockage. 

All we know is that this must be somehow related to our knowledge on the {\em non-constructive} determination of a solution for linear but non-monotonous constraints in a highly non trivial functional space. In fact, it is easy to show that in our case, the solution is unic when it exists, which means that there is still a lot of symmetry that we where unable to use.

}

\newpage

\bibliography{article}

\end{document}